\pdfoutput=1
\documentclass{siamart171218}
\usepackage[T1]{fontenc}
\usepackage[latin9]{inputenc}
\usepackage{array}
\usepackage{./svg}
\usepackage{verbatim}
\usepackage{textcomp}
\usepackage{bm}
\usepackage{amsfonts}
\usepackage{amssymb}
\usepackage{amsmath}
\usepackage{graphicx}
\usepackage{setspace}

\usepackage{graphicx,epstopdf}

\makeatletter

\providecommand{\tabularnewline}{\\}
\newsiamremark{remark}{Remark} 

\makeatother
\usepackage[english]{babel}

\title{Multi-clusters in networks of adaptively coupled phase oscillators\thanks{Submitted to the editors August 31, 2018.
		\funding{This work was supported by the German Research Foundation (DFG) within SCHO 307/15-1 and YA 225/3-1 "DFG-RSF: Complex dynamical networks: effects of heterogeneity, adaptivity, and delays".
	}}}
\author{
		Rico Berner\footnotemark[2] \footnotemark[3] 
		\and Eckehard Schöll\thanks{Institute of Theoretical Physics, Technische Universit\"at Berlin, Hardenbergstr. 36, D-10623 Berlin, Germany, (\email{rico.berner@physik.tu-berlin.de}).}
		\and Serhiy Yanchuk\thanks{Institute of Mathematics, Technische Universit\"at Berlin, Strasse des 17. Juni 136, D-10623 Berlin, Germany.}
}

\headers{Multi-clusters in adaptive networks}{Rico Berner, Eckehard Schöll, and Serhiy Yanchuk}

\begin{document}
\maketitle
\begin{abstract}
Dynamical systems on networks with adaptive couplings appear naturally
in real-world systems such as power grid networks, social networks
as well as neuronal networks. We investigate a paradigmatic system
of adaptively coupled phase oscillators inspired by neuronal networks
with synaptic plasticity. One important behaviour of such systems reveals
splitting of the network into clusters of oscillators with the same
frequencies, where different clusters correspond to different frequencies. Starting
from one-cluster solutions we provide existence criteria for multi-cluster solutions and present their explicit form. The phases of the oscillators within
one cluster can be organized in different patterns: antipodal, double antipodal, and splay type. Interestingly, multi-clusters are shown to exist
where different clusters exhibit different patterns. For instance,
an antipodal cluster can coexist with a splay cluster. We also provide stability conditions for one- and
multi-cluster solutions. These conditions, in particular,
reveal a high level of multistability. 
\end{abstract}



\section{Introduction\label{sec:Introduction}}

Collective behavior in networks of coupled oscillatory systems
has attracted a lot of attention in the last decades. Depending on
the network and the specific dynamical system, various synchronization
patterns of increasing complexity were explored. Even in simple models
of coupled oscillators, patterns such as complete synchronization
\cite{KUR84,PIK01}, cluster synchronization where the network
splits into groups of synchronous elements \cite{DAH12}, or special types of spatial
coexistence of coherent (synchronized) and incoherent (desynchronized)
states have been found~\cite{OME11,ABR04,KUR02a}. Apart from the
theoretical importance, these complex solutions are found in a wide range
of experimental systems including optoelectronic networks \cite{SOR13},
chemical networks \cite{TIN12,HAG12,VAN00}, neural networks \cite{HAM07},
ecological and climate systems \cite{BLA99a}.

The real-world systems are often described by the networks that change
their structure over time. For instance, the synaptic connections
between neurons change depending on relative timings of neuronal spiking
\cite{ABB00,LUE16,MAR97a,CAP08a}. Thus the network structure reorganize
adaptively in response to the dynamics. Similarly, chemical systems
have been reported \cite{JAI01}, where the reaction rates adapt dynamically
depending on the variables of the system. Activity-dependent plasticity
is also common in biological or social systems \cite{GRO08a}. Alternatively, adapting the network topology has also successfully been used in order to control cluster synchronization in delay-coupled networks \cite{LEH14}.

In this article, we consider a network of $N$ coupled phase oscillators
with adaptive coupling
\begin{align}
\frac{d\phi_{i}}{dt} &=\omega-\frac{1}{N}\sum_{j=1}^{N}\kappa_{ij}\sin(\phi_{i}-\phi_{j}+\alpha),\label{eq:PhiDGL_general} \\
\frac{d\kappa_{ij}}{dt}&=-\epsilon\left(\sin(\phi_{i}-\phi_{j}+\beta)+\kappa_{ij}\right),\label{eq:KappaDGL_general}
\end{align}
where $\phi_{i}$ represents the phase of the $i$th
oscillator ($i=1,\dots,N$) and $\omega$ the natural frequency. The
interaction between the phase oscillators is described by the coupling
matrix $\kappa(t):=\left(\kappa_{ij}(t)\right){}_{i,j\in{1,\dots,N}}$,
$\kappa_{ij}(t)\in[-1,1]$. Thus, the phase space of
the system $\left(\phi(t),\kappa(t)\right)\in\mathbb{T}^{N}\times[-1,1]^{N^{2}}$
is $(N+N^{2})$-dimensional with $\mathbb{T}^N$ denoting the $N$-torus. The parameter $\alpha$ can be considered
as a phase-lag of the interaction \cite{SAK86}. System (\ref{eq:PhiDGL_general})\textendash (\ref{eq:KappaDGL_general})
has attracted a lot of attention recently \cite{KAS17,AOK09,AOK11,KAS16a,NEK16,GUS15a,PIC11a,TIM14,REN07,AVA18},
since it is a first choice paradigmatic model for the modeling of
the dynamics of adaptive networks. In particular, it generalizes the
Kuramoto (or Kuramoto-Sakaguchi) model with fixed $\kappa$ \cite{ACE05a,KUR84,OME12b,STR00,PIK08}. 

The matrix $\kappa(t)$ characterizes the coupling topology of the network
at time $t$. Assume that $\epsilon$ is a small but not vanishing parameter. Then, the dynamical equation (\ref{eq:KappaDGL_general}) describes the adaptation of the network topology depending
on the dynamics of the network nodes. Using the terminology from
neuroscience, such an adaptation can be also called plasticity ~\cite{AOK11}.
The chosen adaptation function in the form $\sin(\phi_{i}-\phi_{j}+\beta)$
with control parameter $\beta$ can take into account different plasticity
rules that can occur in neuronal networks, see Fig.~\ref{fig:Plasticity_fucntion_betadep}.
For instance, for $\beta=-\pi/2$, the Hebbian rule is obtained where
the coupling $\kappa_{ij}$ is increasing between any two systems
with close phases, i.e., $\phi_{i}-\phi_{j}$ close to zero ~\cite{HEB49,HOP96,SEL02,AOK15}.
If $\beta=0$ the link $\kappa_{ij}$ will be strengthened if the
$j$-th oscillator is advancing the $i$-th. Such a relationship is
typical for spike-timing-dependent plasticity in neuroscience
\cite{CAP08a,MAI07,LUE16,POP13}.

\begin{figure}
	\centering 
	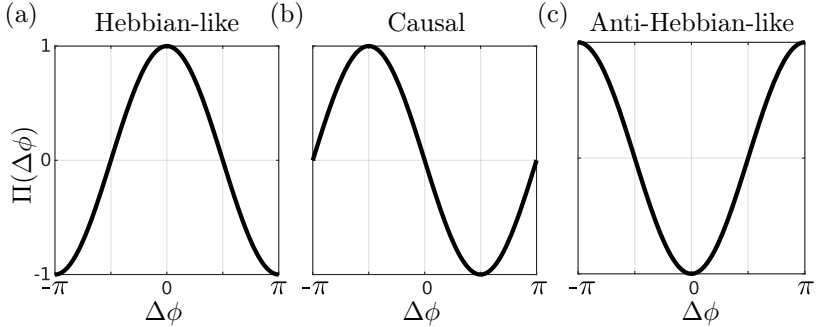
	\caption{The plasticity function $-\sin(\Delta\phi+\beta)$ and corresponding plasticity rules. (a) $\beta=-\frac{\pi}{2}$, (b) $\beta=0$, (c) $\beta=\frac{\pi}{2}.$\label{fig:Plasticity_fucntion_betadep}}
\end{figure}
Let us mention important properties of the model. Firstly, the parameter
$\epsilon\ll1$ separates the time scales of the slow dynamical behaviour
of the coupling strengths and the fast dynamics of the oscillatory
system. Due to the invariance of system (\ref{eq:PhiDGL_general})\textendash (\ref{eq:KappaDGL_general}) with respect to the phase-shift $\phi_i\mapsto\phi_i+\psi$ for all $i=1,\dots,N$ and $\psi\in\mathbb{T}^1$, the frequency $\omega$ can be set to zero in the co-rotating coordinate frame $\phi\mapsto\phi+\omega t$. Moreover, one can restrict the consideration of the coupling weights to the interval $-1\le\kappa_{ij}\le1$ due to the existence of the attracting region $G:=\left\{ \left(\phi_{i},\kappa_{ij}\right):\phi_{i}\in\mathbb{T}^{1},|\kappa_{ij}|\le1,\,i,j=1,\dots,N\right\}$~\cite{KAS17}. 

Finally, let us mention the symmetries of the model
\begin{align*}
	(\alpha,\beta,\phi_{i},\kappa_{ij}) & \mapsto(-\alpha,\pi-\beta,-\phi_{i},\kappa_{ij}),\\
	(\alpha,\beta,\phi_{i},\kappa_{ij}) & \mapsto(\alpha+\pi,\beta+\pi,\phi_{i},-\kappa_{ij}).
\end{align*}
As a result of these symmetries one can restrict the analysis to the parameter region $\alpha\in[0,\pi/2)$ and $\beta\in[-\pi,\pi)$. 

System (\ref{eq:PhiDGL_general})\textendash (\ref{eq:KappaDGL_general})
has been studied numerically in \cite{AOK09,AOK11,NEK16,KAS16a,KAS17}.
In particular, it is shown that starting from uniformly distributed
random initial condition ($\phi_{i}\in[0,2\pi)$, $\kappa_{ij}\in[-1,1]$
for all $i,j\in{1,\dots,N}$) the system can reach different multi-cluster
solutions with hierarchical structure depending on the parameters $\alpha$
and $\beta$. The individual cluster of the multi-clusters consists
of either splay or anti-phase synchronous
solutions. In addition, multi-cluster solutions are reported in an adaptive network of Morris-Lecar bursting neurons with spike timing-dependent plasticity rule~\cite{POP15}.

The goal of this paper is to provide an analytic description of the
one-cluster and multi-cluster solutions. Additionally, we report a new,
mixed multi-cluster solution. These new multi-clusters consist simultaneously
of splay and (in-) anti-phase synchronous solutions. Thus, all one-cluster solutions can serve as building blocks for a multi-cluster solution. We show that multi-clusters corresponding to certain fixed frequency ratio appear in continuous
families, and, moreover, multi-clusters with different frequency ratios
can coexist for the same parameter values. We derive explicit algebraic
equations for the frequencies of coexisting multi-cluster solutions.
In a particular case of two clusters, these equations can be explicitly solved. 

The structure of the paper is as follows. Sections~\ref{sec:One-cluster}
and \ref{sec:Multi-cluster-states} present our main results on the
existence of one- and multi-cluster solutions, respectively. The multi-cluster
solutions are illustrated using numerical integration of (\ref{eq:PhiDGL_general})\textendash (\ref{eq:KappaDGL_general}).
Section~\ref{sec:StabilityAnalysis} investigates stability of these
solutions. Finally, section~\ref{sec:conclusion} summarizes our main results and puts them into the context of general adaptive networks. The proofs of the propositions in sections~\ref{sec:One-cluster}
and \ref{sec:Multi-cluster-states} are presented in Appendix~\ref{sec:Proofs}.

\section{One-cluster solutions\label{sec:One-cluster}}

We start with the collective dynamics, where all oscillators
are synchronous up to phase shifts, i.e. $\phi_{i}=s(t)+a_{i}$. It
is easy to see from~(\ref{eq:PhiDGL_general}), that $ds/dt=\text{const}$
in this case, and, hence, $s(t)=\Omega t$ with some constant frequency
$\Omega$. Moreover, due to the symmetry of system (\ref{eq:PhiDGL_general})\textendash (\ref{eq:KappaDGL_general}) with respect to the phase-shift $\phi_{i}\mapsto\phi_{i}-a_{1}$ one can consider $a_{1}=0$ without loss of generality. Therefore, we define the
following solutions.
\begin{definition}
\label{def:PhaseOscStates} Phase oscillators $\phi_{i}(t)$, $i=1,\dots,N$
are said to be \\
(i) \textbf{in-phase synchronous} if $\phi_{i}(t)=s(t)$ for all $i$;
\\
(ii) \textbf{anti-phase synchronous} if $\phi_{i}(t)=s(t)+a_{i}$ with
$a_{i}\in\{0,\pi\}$ and there are $i\neq j$ such that $a_{i}\ne a_{j}$;\\
(iii) \textbf{rotating-waves} if $\phi_{i}(t)=s(t)+(i-1)2\pi k/N$, where
$k\in\left\{ 1,\dots,N\right\} $ is the wave number;\\
(iv) \textbf{phase-locked} if $\phi_{i}=s(t)+a_{i}$ with arbitrary
$a_{i}\in\mathbb{T}^{1}$.
\end{definition}
Note that the following implications hold: rotating-waves with $k=0$
and $k=N/2$ are in-phase and anti-phase synchronous, respectively.
Rotating-waves, in-phase, and anti-phase solutions are phase-locked.
The term ''rotating-wave'' relates to the rotating symmetry of the
solutions with respect to the spatial coordinate given by the index
$i$. These solutions are also known as twisted states~\cite{WIL06}. 

The following quantity is used for measuring the synchronization.
\begin{definition}
The\textbf{ $\mathbf{n}$th order parameter} for the state $\bm{\phi}\in\mathbb{T}^{N}$
is defined as 
\begin{align}
Z_{n}(\mathbf{\bm{\phi}})=R_{n}(\bm{\phi}) & e^{\mathrm{i}\theta_{n}(\bm{\phi})}:=\frac{1}{N}\sum_{j=1}^{N}e^{\mathrm{i}n\phi_{i}},\label{eq:Nth_OrderParameter}
\end{align}
where $n\in\mathbb{N}$. The symbols $R_{n}>0$ and $\theta_{n}\in\mathbb{T}^{1}$
denote the modulus and the phase of the order parameter, respectively.
\end{definition}
Note, the symbol $\mathrm{i}$ denotes the imaginary unit to distinguish it from the index $i$. It is straightforward to check the values of the $n$th order parameter
for the special solutions defined above. For instance, for an in-phase
synchronous solution it holds $Z_{n}(\bm{\phi}(t))=e^{\mathrm{\mathrm{i}}n\Omega t}$
and, hence $R_{n}=1$ for all $n\in\mathbb{N}$ and $t\in\mathbb{R}$.
For the other solutions the results are summarized in Table~\ref{tab:Nth_OrderParameter_DiffStates}.
If $\bm{\phi}(t)$ is a phase-locked solution, the modulus of the $n$th
order parameter does not depend on $\Omega t$, and, hence, it will
be often referred to as $R_{n}(\mathbf{a})$, where $\mathbf{a}=(a_{1},\dots,a_{N})^{T}$
is the phase shift vector.
\begin{table}
\begin{centering}
\begin{tabular}{|>{\centering}m{5cm}|>{\centering}m{6cm}|}
\hline 
\noalign{\vskip\doublerulesep}
State & $n$th order parameter\tabularnewline[\doublerulesep]
\hline 
\noalign{\vskip\doublerulesep}
\hline 
\noalign{\vskip\doublerulesep}
in-phase

$\phi_{i}=\Omega t$ & $R_{n}=1$\tabularnewline[\doublerulesep]
\noalign{\vskip\doublerulesep}
\hline 
\noalign{\vskip\doublerulesep}
anti-phase

$\phi_{i}=\Omega t+a_{i}$, $a_{i}\in\left\{ 0,\pi\right\} $ & $R_{2n}=1$, $R_{2n+1}=\left|2\frac{N_{1}}{N}-1\right|$, where $N_{1}$
is the number of $a_{i}=0$\tabularnewline[\doublerulesep]
\noalign{\vskip\doublerulesep}
\hline 
\noalign{\vskip\doublerulesep}
rotating-wave

$\phi_{i}=\Omega t+i\frac{2\pi}{N}k$, $k\ne0,\frac{N}{2}$  & $R_{n}=1$ if $\frac{n\cdot k}{N}\in\mathbb{N}_0$ and $R_{n}=0$ otherwise\tabularnewline[\doublerulesep]
\noalign{\vskip\doublerulesep}
\hline 
\noalign{\vskip\doublerulesep}
phase-locked

$\phi_{i}=\Omega t+a_{i}$, $a_{i}\in\mathbb{T}^{1}$ & $R_{n}=\frac{1}{N}\left|\sum_{j=1}^{N}e^{\mathrm{i}na_{j}}\right|$\tabularnewline[\doublerulesep]
\hline 
\noalign{\vskip\doublerulesep}
\end{tabular}
\par\end{centering}
\caption{$n$th order parameter for the phase-locked solutions introduced in
Definition~\ref{def:PhaseOscStates}\label{tab:Nth_OrderParameter_DiffStates}.}
\end{table}
With the help of the order parameter we can introduce the following three types of phase-locked states.
\begin{definition}
	\label{def:PhaseOscStatesSPLAYANTDOUBLE} Phase oscillators $\phi_{i}(t)$, $i=1,\dots,N$
	are said to form a\\
	(i) \textbf{(Splay cluster)} if $R_{2}(\bm{\phi})=0$;\\
	(ii) \textbf{(Antipodal cluster)} if $R_{2}(\bm{\phi})=1$, \textit{i.e.}, $\phi_{i}\in\left\{ 0,\pi\right\}$
	for all $i=1,\dots,N$;
	\\
	(iii) \textbf {(Double antipodal cluster)} if $\phi_{i}\in\left\{ 0,\pi,\psi,\psi+\pi\right\}$ for all $i=1,\dots,N$ with $\psi\in (0,\pi)$.
\end{definition}
Note that if $\bm{\phi}$ is in-phase or anti-phase synchronous, $\bm{\phi}$ forms an antipodal cluster as well. Rotating waves from Definition~\ref{def:PhaseOscStates} are also known as splay states or incoherent clusters~\cite{CHO09}. For all theses rotating waves, it holds $R_1(\bm{\phi})=R_2(\bm{\phi})=0$. As it will be shown in Proposition~\ref{prop:1ClusterSolutions}, system (\ref{eq:PhiDGL_general})\textendash (\ref{eq:KappaDGL_general}) generically possesses solutions with $R_2(\bm{\phi})=0$ rather than $R_1(\bm{\phi})=0$. Both uniformity criteria are clearly related since $R_2(\bm{\phi})=R_1(2\bm{\phi})$. We use the notion splay cluster in order to stress that the phases are uniformly distributed around the unit circle with respect to the second moment of the order parameter.
The following result describes all possible phase-locked solutions in the
system of adaptively coupled oscillators~(\ref{eq:PhiDGL_general})\textendash (\ref{eq:KappaDGL_general}). We call these solutions one-cluster solutions, since all oscillators possess the same frequency.
\begin{proposition}
\label{prop:1ClusterSolutions} System (\ref{eq:PhiDGL_general})\textendash (\ref{eq:KappaDGL_general})
possesses the following phase-locked solutions 
\begin{align}
\phi_{i} & = \Omega t+a_{i},\label{eq:pPL}\\
\kappa_{ij} & = -\sin(a_{i}-a_{j}+\beta),\quad i,j=1,\dots,N\label{eq:kPL}
\end{align}
if and only if one of the following three conditions is fulfilled:\\
(i) the phases $a_{i}$ form a splay cluster, i.e., $R_{2}(\mathbf{a})=0$;\\
(ii) the phases $a_{i}$ form an antipodal cluster, i.e., $R_{2}(\mathbf{a})=1$;
\\
(iii) the phases $a_{i}$ form a double antipodal cluster with $m\in\left\{ 1,\dots,N-1\right\} $, $a_{i}\in\left\{ 0,\pi,\psi_{m},\psi_{m}+\pi\right\}$,
$i=1,\dots,N$ and $\psi_{m}$ being the unique modulo $2\pi$ solution
to the following equation
\begin{align*}
\frac{N-m}{m}\sin(\psi-\alpha-\beta)=\sin(\psi+\alpha+\beta),
\end{align*}
and the number of phase shifts $a_{i}$ such that $a_{i}\in\{0,\pi\}$
equals to $m$. 

The corresponding frequencies are given by
\begin{equation}
\Omega=\begin{cases}
\cos(\alpha-\beta)/2 & \mbox{if }\,R_{2}(\mathbf{a})=0,\\
\sin\alpha\sin\beta & \mbox{if }\,R_{2}(\mathbf{a})=1,\\
\cos(\alpha-\beta)/2-\frac{1}{2}R_{2}(\mathbf{a})\cos(\psi) & \text{in case (iii)}.
\end{cases}\label{eq:OmegaPL}
\end{equation}
\end{proposition}
The proof of this and other propositions are given in App.~\ref{sec:Proofs}. Note that for the special cases $\alpha=0$ or $\alpha=\pi$ and $\beta=\alpha+\pi/2$
or $\beta=\alpha+3\pi/2$ solutions with $R_{2}(\mathbf{a})\notin\left\{ 0,1\right\} $
were discussed in \cite{GUS15a}. Moreover, similar solutions were found
in experimental settings with delay-coupled chemical oscillators~\cite{BLA13}.

Note that conditions (i)\textendash (iii) of Proposition \ref{prop:1ClusterSolutions} imply that there are three possible types of one-cluster solutions: splay, antipodal, and double antipodal. We illustrate these solutions in Figs.~\ref{fig:1Cl_Illustration}(a\textendash c). Without loss of generality, we neglect self-coupling for all numerical simulations. Hence, the entries of the coupling matrix $\kappa$ are zero on the diagonal. Indeed, in the new co-rotating coordinate frame $\phi\mapsto\phi-(1/N)\sin(\alpha)\sin(\beta)t$, the system ~\eqref{eq:PhiDGL_general}--\eqref{eq:KappaDGL_general} with self-coupling is equivalent to the same system without self-coupling.

\begin{figure}
	\begin{center}
		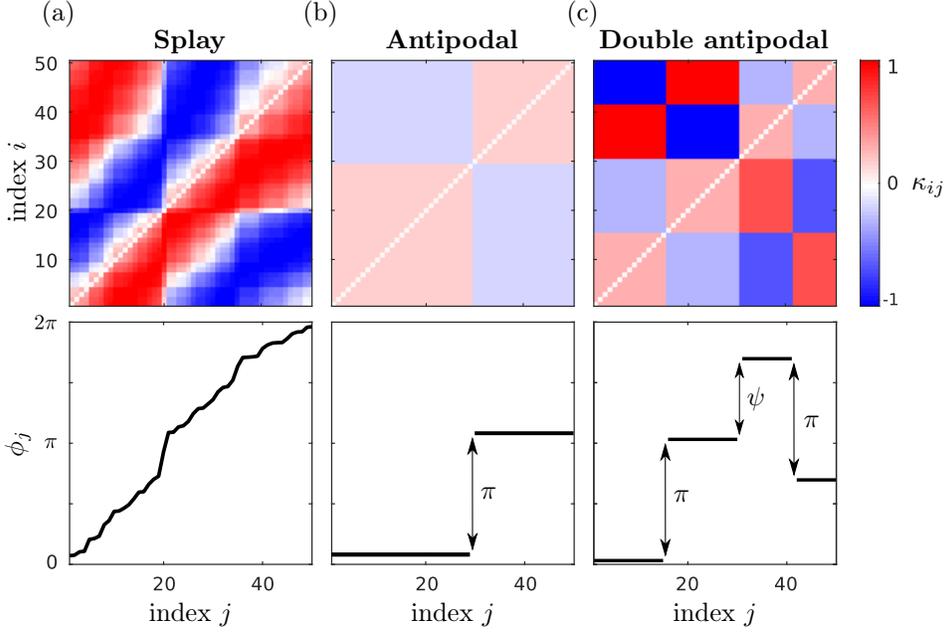
	\end{center}%
\caption{\label{fig:1Cl_Illustration}Illustration of the three types of one-cluster
solutions given by (\ref{eq:pPL})\textendash (\ref{eq:kPL}) for an
ensemble of $50$ oscillators. One-cluster solutions (a) of splay type ($R_{2}(\mathbf{a})=0$)
for $\alpha=0.3\pi$, $\beta=0.1\pi$, (b) of antipodal type ($R_{2}(\mathbf{a})=1$),
for $\alpha=0.2\pi$, $\beta=-0.95\pi$ and (c) of double antipodal type satisfying
condition (iii) of Proposition \ref{prop:1ClusterSolutions} with
$m=30$ for $\alpha=0.3\pi$, $\beta=-0.15\pi$.
}
\end{figure}

In the following, we will focus our study on the splay and antipodal clusters with $R_{2}(\mathbf{a})\in\{0,1\}$,
i.e. the phase-locked solutions given by cases (i) and (ii) of Proposition
\ref{prop:1ClusterSolutions}. We further remark that the phase-locked solutions (\ref{eq:pPL})\textendash (\ref{eq:kPL}) are relative equilibria with respect to the phase-shift defined in section~\ref{sec:Introduction}, \textit{i.e.}, they are equilibria in the co-rotating frame $\phi\mapsto\phi+\Omega t$.

If $R_{2}(\mathbf{a})=1$, Proposition \ref{prop:1ClusterSolutions}
implies that $a_{i}$ are either $0$ or $\pi$. Therefore, there
are $2^{N-1}$ isolated solutions of this kind. Note that the in-phase synchronous solution is an antipodal one-cluster solution.

The situation is different for the splay cluster. The relation $R_{2}(\mathbf{a})=0$ gives
the $N-2$ parametric ($N>2$) family 
\begin{multline} \label{eq:S}
S:=\Bigl\{(\phi_{i},\kappa_{ij}):\phi_{i}=\Omega t+a_{i},\,\,\kappa_{ij}=-\sin(a_{i}-a_{j}+\beta),\\
\sum_{j=1}^{N}\sin(2a_{j})=\sum_{j=1}^{N}\cos(2a_{j})=0\Bigr\},
\end{multline}
where $\Omega=\cos(\alpha-\beta)/2$. Moreover, analogously to~\cite{ASH08}, one can show that $S$ is the union of $N-2$ dimensional manifolds.

The structure of the solution family (\ref{eq:S}) is illustrated in Figs.~\ref{fig:1Cl_type1_manifold}(a\textendash c)
for $N=2,3,4$.
\begin{figure}
	\begin{center}
		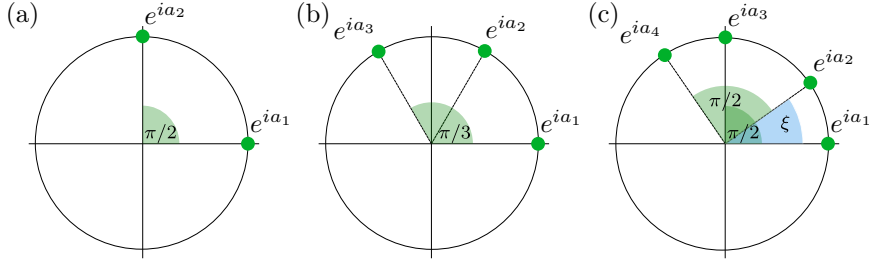
	\end{center}
\caption{\label{fig:1Cl_type1_manifold}Illustration of the family of solutions
$S$ (a) $N=2$, (b) $N=3$, (c) $N=4$.}
\end{figure}
Figure~\ref{fig:1Cl_type1_manifold}(a) shows one of the two disjoint
one-dimensional subsets of $S$ for the case of two adaptively coupled
oscillators modulo common rotation of the phases on the circle. In
fact, the oscillators have to have a phase shift of $\pi/2$ in order
to meet the condition $R_{2}(\mathbf{a})=0$, i.e. $\mathbf{a}=\left(\gamma,\gamma+\pi/2\right)$,
$\gamma\in[0,2\pi)$. The dimension of $S$ for $N=2$ is one. For a system consisting of three or four phase oscillators
the dimension of $S$ is either $1$ or $2$, respectively. For $N=3$,
one has $\mathbf{a}=\left(\gamma,\gamma+\pi/3,\gamma+2\pi/3\right)$,
see Fig.~\ref{fig:1Cl_type1_manifold}(b). For the case $N=4$, we
have $\mathbf{a}=(\gamma,\gamma+\xi,\gamma+\pi/2,\gamma+\xi+\pi/2),$
$\gamma,\xi\in[0,2\pi)$, see Fig.~\ref{fig:1Cl_type1_manifold}(c).

Note that the set of phases satisfying the condition $R_1(\mathbf{a})=0$ was described in~\cite{ASH08,BUR11,ASH16a}. Our case of splay clusters can be related to this set using the equality $R_2(\mathbf{a})=R_1(2\mathbf{a})=0$.

Rotating-waves are a particular case of the splay cluster, namely, the following corollary holds. 
\begin{corollary}
\label{cor:1ClusterRWSol}For any $k\in\left\{ 1,\dots,N-1\right\} ,$ $k\ne N/2$
the rotating-wave 
\begin{align*}
\phi_{i} & = \Omega t+(i-1)\frac{2\pi k}{N},\\
\kappa_{ij} & = -\sin\left((i-j)\frac{2\pi k}{N}+\beta\right),
\end{align*}
with $\Omega=\cos(\alpha-\beta)/2$ is a solution of system~(\ref{eq:PhiDGL_general})\textendash (\ref{eq:KappaDGL_general}).
\end{corollary}

Let us make a short remark, which allows for a better understanding and
interpretation of the phase-locked solutions given in Proposition~\ref{prop:1ClusterSolutions}.
Assume that the phase variables are in a phase-locked solution $\phi_{i}=\Omega t+a_{i}$. Then,
the coupling weights $\kappa_{ij}$ have to satisfy the linear system 
\begin{align}
\dot{\kappa}_{ij} & =-\epsilon\kappa_{ij}-\epsilon\sin(a_{i}-a_{j}+\beta).\label{eq:KappaDGL}
\end{align}
This system has the unique solution $\kappa_{ij}=-\sin(a_{i}-a_{j}+\beta)$ which is constant, bounded on $\mathbb{R}$, and asymptotically stable as $t\to\infty$. Therefore, the specific network connectivity $\kappa_{ij}=-\sin(a_{i}-a_{j}+\beta)$ is associated with a given phase-shift $\mathbf{a}$.

\section{Multi-cluster solutions\label{sec:Multi-cluster-states}}

As previous numerical studies \cite{KAS17} found out, the phase-locked
solutions described in section~\ref{sec:One-cluster} can act as building
blocks for (hierarchical) multi-cluster solutions. 
\begin{definition}
\label{def:PhaseOscStates_MulCl} Phase oscillators $\phi_{i}(t)$
form a \textbf{multi-cluster} if they can be separated
into $M$ groups of phase-locked oscillators (clusters), i.e., for all $\mu\in\{1,\dots,M\}$
the phase oscillators $\phi_{i,\mu}$, $i\in\{1,\dots,N_{\mu}\}$,
from each group $\mu$ satisfy $\phi_{i,\mu}(t)=s_{\mu}(t)+a_{i,\mu}$. 
\end{definition}
The appearance of multi-clusters
is interesting and nontrivial, since such solutions, in contrast to
one-clusters, are no more relative equilibria of (\ref{eq:PhiDGL_general})\textendash (\ref{eq:KappaDGL_general}),
but are periodic or quasi-periodic solutions, which appear due to
the special structure of the equation and adaptive nature of the coupling. The oscillators within one cluster posses a synchronized
temporal dynamics with possible phase lags. In a multi-cluster, the coupling matrix $\kappa$
can be divided into different blocks according to the division by
clusters: $k_{ij,\mu\nu}$ will refer to the coupling weight between
the $i$-th oscillator of cluster $\mu$ to the $j$-th oscillator
of cluster $\nu$.

Depending on the type of the constituting individual clusters, different
multi-clusters may be observed: splay, antipodal and mixed
type. The following sections describe each particular multi-cluster
solution.

\subsection{Multi-cluster solutions of splay type\label{subsec:MC_RotWave}}

The multi-cluster solutions of splay type are composed by the
clusters from the continuous family $S$ of phase-locked solutions with
$R_{2}(\mathbf{a})=0$ and different frequencies. The following proposition
describes them.
\begin{proposition}
\label{prop:MCSol_RotWave}System (\ref{eq:PhiDGL_general})\textendash (\ref{eq:KappaDGL_general})
possesses the multi-cluster solution 
\begin{align}
\phi_{i,\mu}(t) & =\Omega_{\mu}t+a_{i,\mu}, &  & \begin{split} i=1,\dots,N_{\mu}\\ \mu =1,\dots,M\end{split} \label{eq:pRWMC}\\
\kappa_{ij,\mu\nu}(t) & =-\rho_{\mu\nu}\sin(\Delta\Omega_{\mu\nu}t+a_{i,\mu}-a_{j,\nu}+\beta-\psi_{\mu\nu}), &  & \begin{split} j=1,\dots,N_{\nu}\\ \nu =1,\dots,M\end{split}\label{eq:kRWMC}
\end{align}
with pairwise different frequencies $\Omega_{\mu}$, $\Delta\Omega_{\mu\nu}:=\Omega_{\mu}-\Omega_{\nu}$,
$\rho_{\mu\nu}:=\left(1+\left(\Delta\Omega_{\mu\nu}/\epsilon\right)^{2}\right)^{-\frac12}$
and $\psi_{\mu\nu}:=\arctan(\Delta\Omega_{\mu\nu}/\epsilon)$ if and
only if 

$R_{2}(\mathbf{a}_{\mu})=0$ for all $\mu=1,\dots,M$ and the frequencies
$(\Omega_{1},\dots,\Omega_{M})$ solve the following system of equations
\begin{align}
\Omega_{\mu}=\frac{1}{2N}\sum_{\nu=1}^{M}\rho_{\mu\nu}N_{\nu}\cos(\alpha-\beta+\psi_{\mu\nu}), &  & \mu=1,\dots,M.\label{eq:MCRotWave_Omega}
\end{align}
Note that $\rho_{\mu\nu}$ and $\psi_{\mu\nu}$ are functions of $\Omega_{\mu}-\Omega_{\nu}$.
\end{proposition}
Similarly to the one-cluster case, multi-cluster solutions of splay
type give rise to a $(N-2M-1)$-dimensional manifold of solutions
\begin{multline*}
S_{M}:=\Bigl\{(\phi_{i,\mu},\kappa_{ij,\mu\nu}):(\phi_{i,\mu},\kappa_{ij,\mu\nu})\text{ as in \eqref{eq:pRWMC}-\eqref{eq:kRWMC},}\\ R_{2}(\mathbf{a}_{\mu})=0\text{ for all }\mu=1,\dots,M\Bigr\}.
\end{multline*}

Let us remark that the collective frequencies~\eqref{eq:MCRotWave_Omega} are only defined up to a constant due to the phase-shift symmetry of system (\ref{eq:PhiDGL_general})\textendash (\ref{eq:KappaDGL_general}) while the frequency difference is unaffected. An example of a 3-cluster solution of splay type is shown in Figure~\ref{fig:3Cl_RotWave}.
The solution was obtained by integrating system (\ref{eq:PhiDGL_general})\textendash (\ref{eq:KappaDGL_general})
numerically starting from random initial conditions. After sufficiently
long transient time, the order of the oscillators is given by first
sorting the oscillators with respect to their average frequencies.
After that the oscillators with the same frequency are sorted by their
phases. It can be seen from the pictures that the sizes of the three
clusters $N_{\mu}$ ($\mu=1,2,3$) possesses a hierarchical structure,
\textit{i.e.}, $N_{3}< N_{2}< N_{1}$. The coupling strengths between oscillators of the same cluster
vary in a larger range than between those of different clusters. The
coupling between different clusters scales with $\epsilon$ since
$\rho_{\mu\nu}=\epsilon/\Delta\Omega_{\mu\nu}+\mathcal{O}\left(\left(\epsilon/\Delta\Omega_{\mu\nu}\right)^{3}\right)$
and is thus close to zero (uncoupled). The oscillators of the same
cluster evolve in time with the same frequencies $\dot{\phi}_{i,\mu}=\Omega_{\mu}$,
$i=1,\dots,N_{\mu}$. 
\begin{figure}
	\begin{center}
		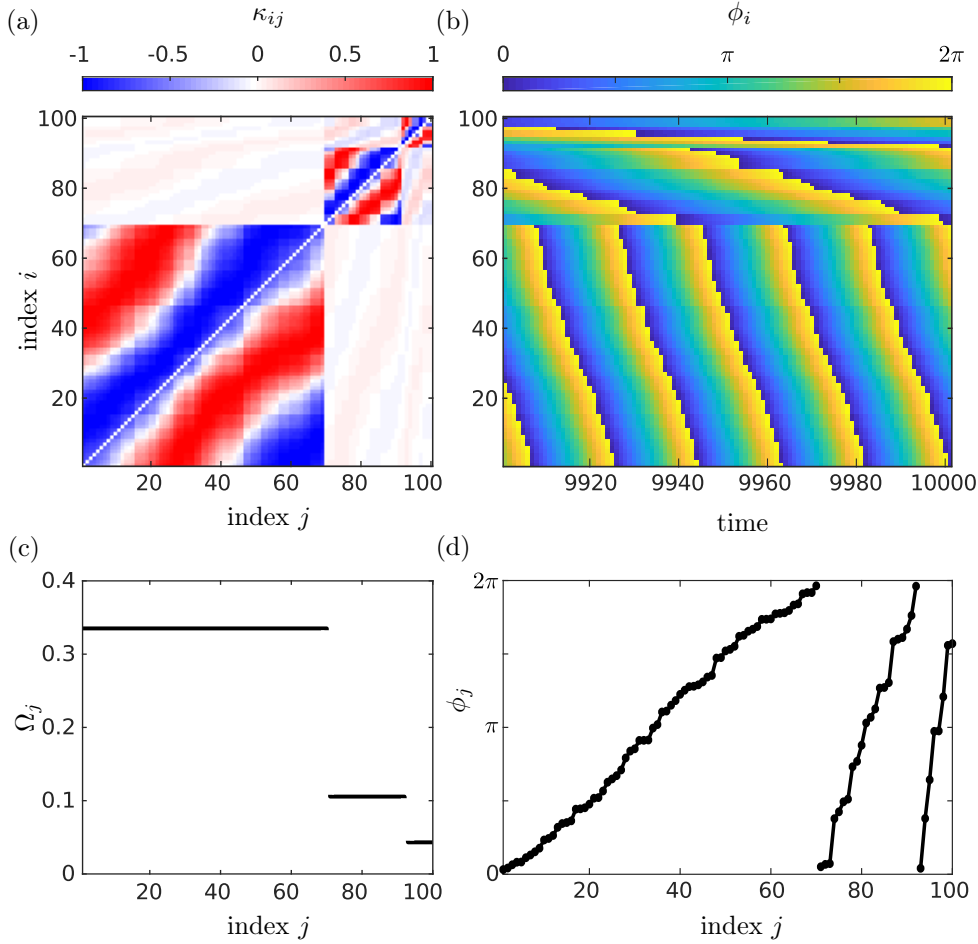
	\end{center}
\caption{Three-cluster of splay type. (a) Coupling weights at $t=10000$
showing three clusters; (b) Distribution of the phases within each
cluster; space-time raster plot; (c) Average frequency of oscillators;
each plateau corresponds to one cluster; (d) Oscillator phases $\phi_{i}(t)$
at fixed time $t=10000$. Parameter values: $\epsilon=0.01$, $\alpha=0.3\pi$,
$\beta=0.23\pi$, and $N=100$. \label{fig:3Cl_RotWave}}
\end{figure}

Let us consider the case of two-clusters in more details. Let $\phi_{i,\mu}$
$\left(\mu=1,2\right)$ with $N_{1}$ and $N_{2}$ being the numbers
of oscillators in cluster $1$ and $2$, respectively. The following
result follows from the Proposition~\ref{prop:MCSol_RotWave}.
\begin{corollary}
Suppose $R_{2}(\mathbf{a}_{\mu})=0$ for $\mu=1,2$, then
\begin{align*}
\phi_{i,1} & =\Omega_{1}t+a_{i,1}, & i=1,\dots,N_{1}\\
\phi_{i,2} & =\Omega_{2}t+a_{i,2}, & i=1,\dots,N_{2}\\
\kappa_{ij,\mu\mu} & =-\sin(a_{i,\mu}-a_{j,\mu}+\beta), & \mu=1,2\\
\kappa_{ij,\mu\nu} & =-\rho_{\mu\nu}\sin(\Delta\Omega_{\mu\nu}t+a_{i,\mu}-a_{j,\nu}+\beta-\psi_{\mu\nu}), & \mu,\nu=1,2;\mu\ne\nu
\end{align*}
is a two-cluster solution of system (\ref{eq:PhiDGL_general})\textendash (\ref{eq:KappaDGL_general})
with
\begin{multline}
\left(\Delta\Omega_{12}\right)_{1,2}=\frac{1}{2}\left(n_{1}-\frac{1}{2}\right)\cos(\alpha-\beta)\\
\pm\frac{1}{2}\sqrt{\left(n_{1}-\frac{1}{2}\right)^{2}\cos^{2}(\alpha-\beta)-2\epsilon(2\epsilon+\sin(\alpha-\beta))},\label{eq:2Cluster_OmegaDiff}
\end{multline}
\begin{equation}
\Omega_{\mu}=\frac{1}{2}\left(n_{\text{\ensuremath{\mu}}}\cos(\alpha-\beta)+\rho_{\mu\nu}n_{\nu}\cos(\alpha-\beta+\psi_{\mu\nu})\right),(\mu,\nu=1,2;\mu\ne\nu),\label{eq:2ClusterOm}
\end{equation}
where $n_{\mu}=N_{\mu}/N$ and $\psi_{\mu\nu}$, $\rho_{\mu\nu}$
as in Proposition~\ref{prop:MCSol_RotWave}. 
\end{corollary}
The explicit expressions for the frequencies $\Delta\Omega_{12}$,
$\Omega_{1,2}$ and other parameters of the solutions follow from
the system of equations (\ref{eq:MCRotWave_Omega}), which can be
solved explicitly leading to (\ref{eq:2Cluster_OmegaDiff})\textendash (\ref{eq:2ClusterOm})
for $M=2$. For any given relative cluster size $n_{\mu}$, equations~(\ref{eq:2Cluster_OmegaDiff})
and (\ref{eq:2ClusterOm}) provide either two, one, or no solutions
corresponding to the two-cluster solution. Hence, for each fixed set of
parameters, there may be up to $2(N-4)$ such solutions. 

Figure~ \ref{fig:2Cl_RotWave_OmegaSols} shows the frequency differences
$\Delta\Omega_{12}$ of these solutions as functions of parameter $\beta$
for different number of oscillations $N$ and adaptation parameters
$\epsilon$. Interestingly, the frequencies of the solutions depend only
on the difference $\alpha-\beta$, see (\ref{eq:2Cluster_OmegaDiff})\textendash (\ref{eq:2ClusterOm}).
By increasing the number of oscillators $N$ in the system, the number
of solutions increases accordingly. This can be seen from Fig.~\ref{fig:2Cl_RotWave_OmegaSols}(a\textendash b)
where we increase the number of oscillators from $N=20$ to $N=50$
with all other parameters fixed. The set of 2-cluster solutions is represented by all $\Delta\Omega_{12}(\beta)$ for a given parameter $\beta$. In accordance with (\ref{eq:2Cluster_OmegaDiff}) the number of solutions increases with increasing
$N$. The region of non-existence of the multi-cluster solutions corresponds
to the cases where the argument beneath the root in~(\ref{eq:2Cluster_OmegaDiff})
becomes negative. The size of the existence gap depends furthermore
on the choice of the time separation parameter $\epsilon$. This can
be seen by comparing Fig.~\ref{fig:2Cl_RotWave_OmegaSols}(b\textendash d)
where we vary the value for $\epsilon.$

\begin{figure}
	\begin{center}
		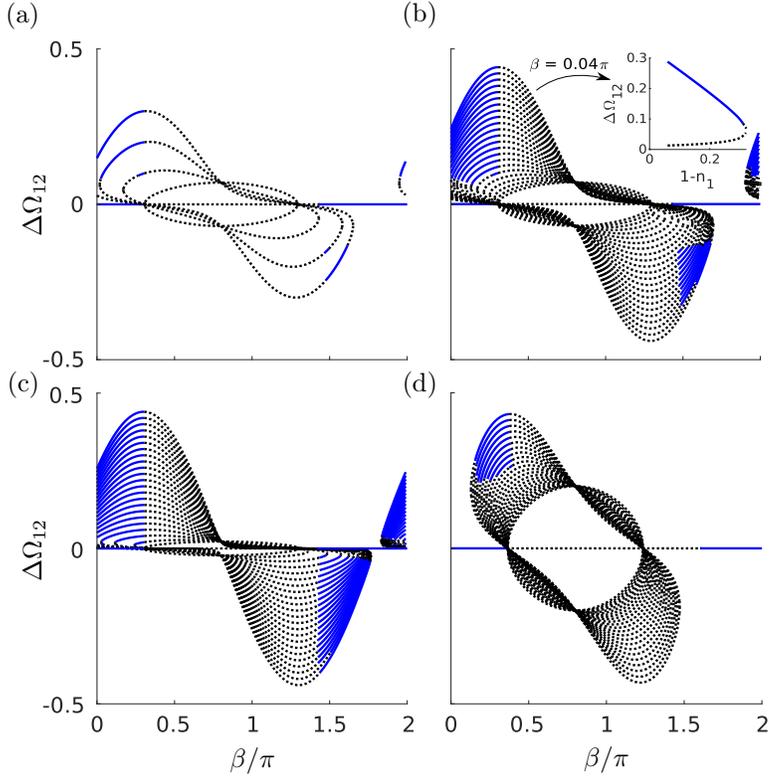
	\end{center}%
\caption{\label{fig:2Cl_RotWave_OmegaSols}The figures show all one- and two-cluster
solutions of splay type for the system (\ref{eq:PhiDGL_general})\textendash (\ref{eq:KappaDGL_general}).
For this, the frequency differences $\Delta\Omega_{12}$ are displayed
corresponding to the equations~(\ref{eq:OmegaPL}) and (\ref{eq:2Cluster_OmegaDiff}).
The dotted lines (black) indicate unstable solutions while the solid
lines (blue) indicate stable solutions. Here, every second solution is
plotted for the sake of visibility. Parameter values: (a) $N=20$,
$\epsilon=0.01$; (b) $N=50$, $\epsilon=0.01$; (c) $N=50$, $\epsilon=0.001$;
(d) $N=50$, $\epsilon=0.1$; $\alpha=0.3\pi$ is fixed for all
panels.}
\end{figure}

\subsection{Multi-cluster solutions of antipodal type\label{subsec:MC_AntiPhase}}

In the case when the oscillators are phase synchronized or in an anti-phase
relation within each cluster,\textit{ }the situation is different
to what was described before. Particularly, the linear growth of the
oscillator phases within each cluster is modulated by periodic or
quasi-periodic terms of order $\epsilon$. That is, the clusters
possess the form $\phi_{i,\mu}(t)=\Omega_{\mu}t+a_{i,\mu}+\epsilon p_{\mu}(t,\epsilon)$.
Here, we give important necessary conditions for the existence of
such solutions and their asymptotic expansion in $\epsilon$. Additionally, we provide numerical results showing these solutions. In particular, we present
a system of equations for the cluster frequencies $\Omega_{\mu}$. 
\begin{proposition}
\label{prop:MCSol_AP_FormalExpansion} Suppose $2a_{i,\mu}=a_{\mu}\,\text{mod\,\,}2\pi$
for all $\mu=1,\dots,M$ and $i=1,\dots,N_{\mu}$. If system (\ref{eq:PhiDGL_general})\textendash (\ref{eq:KappaDGL_general})
possesses antipodal multi-cluster solution $(\phi_{i,\mu},\kappa_{ij,\mu\nu})$
then its first asymptotic expansion in $\epsilon$ is given by
\begin{align*}
\phi_{i,\mu}^{(1)} & =\Omega^{(1)}_{\mu}t+a_{i,\mu}-\epsilon\sum_{\overset{\nu=1}{\nu\ne\mu}}^{M}\frac{n_{\nu}}{4\left(\Delta\Omega^{(1)}_{\mu\nu}\right)^{2}}\cos(2\Delta\Omega^{(1)}_{\mu\nu}t+a_{\mu}-a_{\nu}+\alpha+\beta),\\
\kappa_{ij,\mu\mu}^{(1)} & =-\sin(a_{i,\mu}-a_{j,\mu}+\beta),\\
\kappa_{ij,\mu\nu}^{(1)} & =\frac{\epsilon}{\Delta\Omega^{(1)}_{\mu\nu}}\cos(\Delta\Omega^{(1)}_{\mu\nu}t+a_{i,\mu}-a_{j,\nu}+\beta),\quad\mu\ne\nu
\end{align*}
 with the cluster frequencies $\Omega^{(1)}_{\mu}$ up to first order in
$\epsilon$ whenever the following implicit equation can be solved
\begin{align}
\Omega^{(1)}_{\mu}=\left(n_{\mu}\sin(\alpha)\sin(\beta)-\epsilon\sum_{\overset{\nu=1}{\nu\ne\mu}}^{M}\frac{n_{\nu}}{2\Delta\Omega^{(1)}_{\mu\nu}}\sin(\alpha-\beta)\right).\label{eq:MCAntiPhase_Omega}
\end{align}
 Here, $\mu=1,\dots,M,$ $i,j=1,\dots,N_{\mu}$ and $\Delta\Omega_{\mu\nu}:=\Omega^{(1)}_{\mu}-\Omega^{(1)}_{\nu}$.
\end{proposition}
This first order perturbation reveals a nonlinear modulation $p_{\mu}$,
which is periodic or quasi-periodic with the frequencies $\Delta\Omega^{(1)}_{\mu\nu}$
given by the differences in the frequencies of the clusters. 

Figure~\ref{fig:3Cl_AntiPhase} shows the numerically obtained 3-cluster
solution of antipodal type. The dynamics of system (\ref{eq:PhiDGL_general})\textendash (\ref{eq:KappaDGL_general})
is shown after a sufficiently long transient so that it represents
dynamically stable solution (more on stability in section~\ref{sec:StabilityAnalysis}).
One can clearly observe three clusters in the coupling
matrix. Similarly to the previous multi-cluster cases, we first sort
the oscillators with respect to their average frequency and subsequently
by their phases. In contrast to the splayed distribution
of the phases described in section~\ref{subsec:MC_RotWave}, the oscillators
within the clusters additionally form two groups, in which the phases
differ by $\pi$.
\begin{figure}
	\begin{center}
		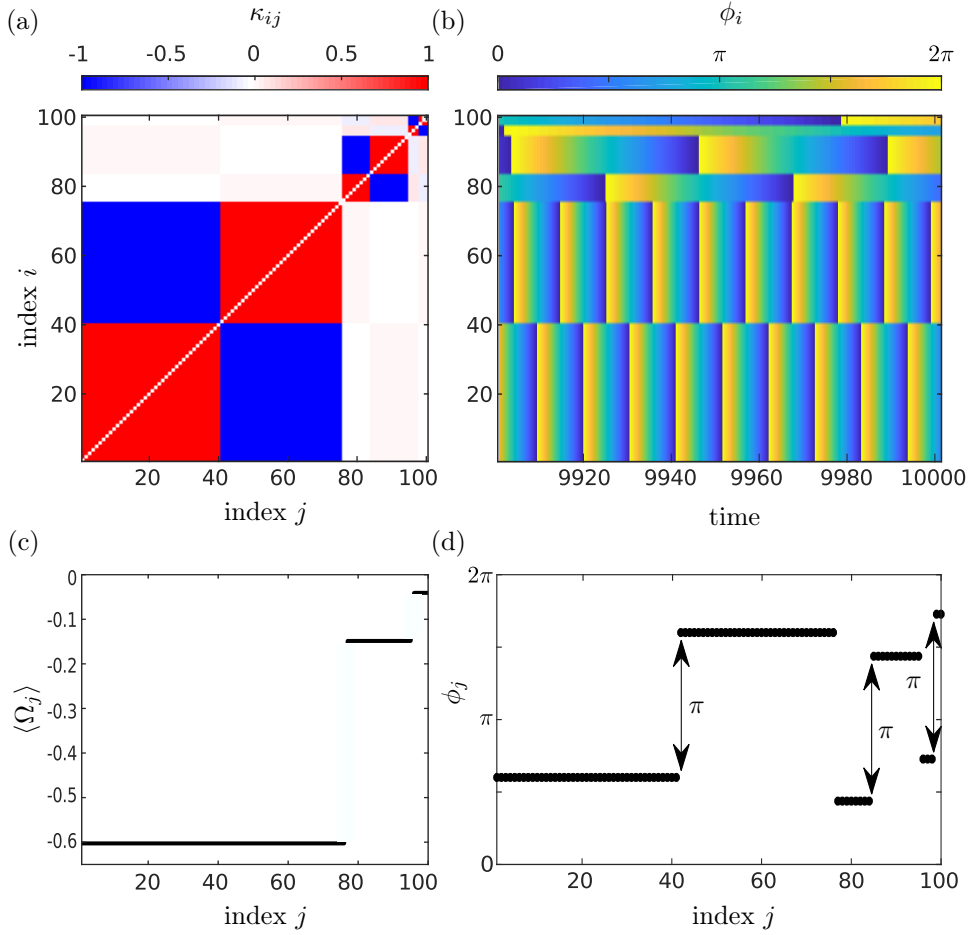
	\end{center}
\caption{Three-cluster of antipodal type. (a) Coupling weights at $t=10000$
showing three clusters; (b) Distribution of the phases within each
cluster; space-time raster plot; (c) Average frequency of oscillators;
each plateau corresponds to one cluster; (d) Oscillator phases $\phi_{i}(t)$
at fixed time $t=10000$. Parameter values: $\epsilon=0.01$, $\alpha=0.3\pi$,
$\beta=-0.53\pi$, and $N=100$. \label{fig:3Cl_AntiPhase}}
\end{figure}

In order to observe the modulation of the cluster frequencies, time series
for three representative oscillators from each cluster are shown in
Fig.~\ref{fig:3Cl_AntiPhase_Fourier_analysis}(a). The averaged linear growth
of the phases due to $\langle\Omega_{\mu}\rangle t$ has been subtracted to show
the modulation. Such small but non-vanishing oscillations do not exist in the case of splay type multi-clusters. In addition, the black dashed lines show the modulation given by Proposition~\ref{prop:MCSol_AP_FormalExpansion} confirming that the asymptotic expansion describes the whole temporal behaviour very well. Furthermore, Proposition~\ref{prop:MCSol_AP_FormalExpansion}
implies that the amplitudes of the modulations are proportional to
$n_{\nu}/(\Delta\Omega^{(1)}_{\mu\nu})^{2}$. Thus, if the difference in the
frequencies is high the amplitude is small and vice versa. This relation
is also reflected by the power spectrum, see Fig.~\ref{fig:3Cl_AntiPhase_Fourier_analysis}(b).
Figure~\ref{fig:3Cl_AntiPhase_Fourier_analysis}(b) confirms that
the frequencies of the modulation oscillations correspond to the differences
of the average frequencies.
\begin{figure}
	\begin{center}
		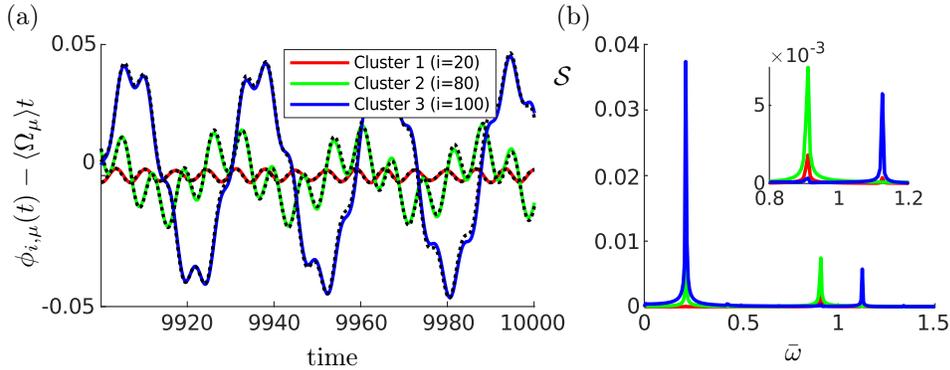
	\end{center}%
\caption{For 3-cluster solution from Fig.\ \ref{fig:3Cl_AntiPhase}, panel (a)
shows time series of an oscillator from one of the clusters after
subtracting the average linear growth $\phi_{\mu,i}(t)-\langle\Omega_{\mu}\rangle t$. The black dashed lines show the corresponding analytic results from the asymptotic expansion in Proposition~\ref{prop:MCSol_AP_FormalExpansion}.
(b) Power spectrum of the time series given in (a).\label{fig:3Cl_AntiPhase_Fourier_analysis}}
\end{figure}

Let us consider the case of two clusters in more details. Let $N_{1}$
and $N_{2}=N-N_{1}$ be the numbers of oscillators in group $1$ and
$2$, respectively. The following result follows from Proposition~\ref{prop:MCSol_AP_FormalExpansion}.
\begin{corollary}
\label{cor:2Cl_AP_FormalExpan}Suppose $2a_{i,\mu}=a_{\mu}$
for all $\mu=1,2$ and $i=1,\dots,N_{\mu}$. If system (\ref{eq:PhiDGL_general})\textendash (\ref{eq:KappaDGL_general})
possesses an antipodal multi-cluster solution $(\phi_{i,\mu},\kappa_{ij,\mu\nu})$
then its first order asymptotic expansion in $\epsilon$ is given
by 
\begin{align*}
\phi_{i,1}^{(1)} & =\Omega^{(1)}_{1}t+a_{i,1}-\epsilon\frac{n_{2}}{4\left(\Delta\Omega^{(1)}_{12}\right)^{2}}\cos(2\Delta\Omega^{(1)}_{12}t+a_{1}-a_{2}+\alpha+\beta),\\
\phi_{i,2}^{(1)} & =\Omega^{(1)}_{2}t+a_{i,2}-\epsilon\frac{n_{1}}{4\left(\Delta\Omega^{(1)}_{12}\right)^{2}}\cos(2\Delta\Omega^{(1)}_{12}t+a_{1}-a_{2}-\alpha-\beta),\\
\kappa^{(1)}_{ij,\mu\mu} & =-\sin(a_{i,\mu}-a_{j,\mu}+\beta),\\
\kappa^{(1)}_{ij,\mu\nu} & =\frac{\epsilon}{\Delta\Omega^{(1)}_{\mu\nu}}\cos(\Delta\Omega^{(1)}_{\mu\nu}t+a_{i,\mu}-a_{j,\nu}+\beta),
\end{align*}
with
\begin{align*}
\Omega^{(1)}_{\mu}=\left(n_{\mu}\sin(\alpha)\sin(\beta)-\epsilon\frac{n_{\nu}}{2\Delta\Omega^{(1)}_{\mu\nu}}\sin(\alpha-\beta)\right) 
\end{align*}
for $\mu=1,2$, $\nu\ne\mu$, $i,j=1,\dots,N_{\mu}$, $\Delta\Omega^{(1)}_{\mu\nu}:=\Omega^{(1)}_{\mu}-\Omega^{(1)}_{\nu}$,
\begin{align}
\left(\Delta\Omega^{(1)}_{12}\right)_{1,2} & =\left(n_{1}-\frac{1}{2}\right)\sin(\alpha)\sin(\beta)\pm\sqrt{\left(n_{1}-\frac{1}{2}\right)^{2}\sin^{2}\alpha\sin^{2}\beta-\frac{\epsilon}{2}\sin(\alpha-\beta)}.\label{eq:2ClusterAntiPhase_OmegaDiff}
\end{align}
\end{corollary}
This result follows directly from Proposition~\ref{prop:MCSol_AP_FormalExpansion}.
It shows, in particular, that the system of equations~(\ref{eq:MCAntiPhase_Omega})
can be solved explicitly by~(\ref{eq:2ClusterAntiPhase_OmegaDiff})
in case of two clusters. 

Similarly to the splay multi-clusters, for any fixed set of parameters
and each $n_{1},$ equation~(\ref{eq:2ClusterAntiPhase_OmegaDiff}) can
lead to two antipodal multi-clusters with two different frequency differences. Hence,
a large number of antipodal two-clusters can coexist for the same parameter
values. Figure~\ref{fig:2Cl_AntiPhase_OmegaSols} illustrates such
a coexistence, where we present the one-cluster solutions given by (\ref{eq:OmegaPL}) and the solutions to the equation~(\ref{eq:2ClusterAntiPhase_OmegaDiff}). Blue solid lines represent those solutions for which the asymptotic expansion led to an existing and stable two-cluster solutions of antipodal type. Note further that for two-cluster solutions of antipodal type,
the asymptotic expansion presented in Proposition~\ref{prop:MCSol_AP_FormalExpansion} turns into a formal expansion whenever $|\Delta\Omega|>\epsilon$, i.e., $\epsilon$ is not assumed to be infinitesimal ($\epsilon \to 0$). The interval $[-\epsilon,\epsilon]$
is therefore highlighted in Fig.~\ref{fig:2Cl_AntiPhase_OmegaSols}. 
\begin{figure}
	\begin{center}
		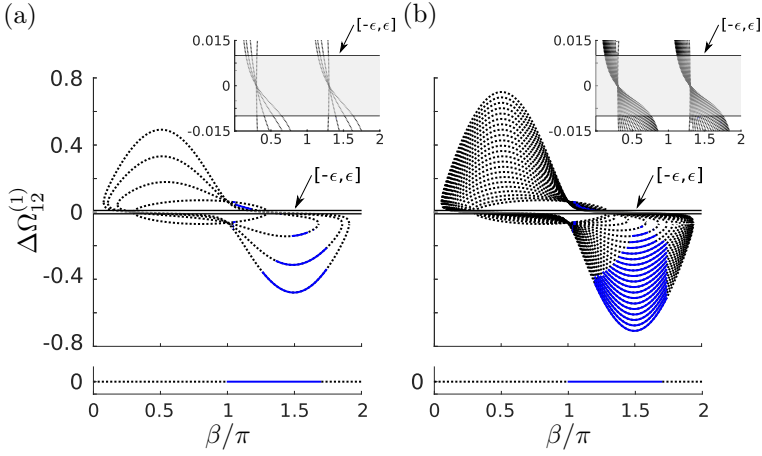
	\end{center}%
\caption{\label{fig:2Cl_AntiPhase_OmegaSols} Two-cluster solutions (upper panels) and one-cluster solutions (lower panels) of antipodal type given by the asymptotic expansion in Corollary~\ref{cor:2Cl_AP_FormalExpan} and Proposition~\ref{prop:1ClusterSolutions}, respectively.
For this, the difference of the frequencies $\Delta\Omega^{(1)}_{12}$ is
displayed corresponding to~(\ref{eq:2ClusterAntiPhase_OmegaDiff}) and~(\ref{eq:OmegaPL}).
The dotted lines (black) indicate unstable solutions while the solid
lines (blue) indicate stable solutions. Here, every second solution is
plotted for the sake of visibility. The insets show a blow-up of the interval $[-\epsilon,\epsilon]$. Parameter values: (a) $N=20$,
$\epsilon=0.01$; (b) $N=50$, $\epsilon=0.01$; $\alpha=0.3\pi$ is fixed for all panels.}
\end{figure}

\subsection{Multi-cluster solutions of mixed type\label{subsec:MC_Mix}}

We have seen how clusters are described consisting of oscillator groups
of splay type (section~\ref{subsec:MC_RotWave}) as well as
clusters consisting of oscillator groups with in- and anti-phase relation
(section~\ref{subsec:MC_AntiPhase}). It is therefore reasonable to
ask for multi-cluster solutions that consist of both of these types.
In order to describe these solutions we have to loosen the definition
of a multi-cluster solution.
\begin{definition}
\label{def:PhaseOscStates_PseudoMulCl} Phase oscillators $\phi_{i}(t)$
form a \textbf{pseudo multi-cluster} if they can
be separated into $M$ groups such that for all $\mu\in\{1,\dots,M\}$
the phase oscillators $\phi_{i,\mu}$, $i\in\{1,\dots,N_{\mu}\}$,
from each group $\mu$ satisfy $\phi_{i,\mu}(t)=\Omega_{\mu}t+s_{i,\mu}(t)$
with bounded functions $s_{i,\mu}$.
\end{definition}
Note that every multi-cluster solution is by definition already a pseudo
multi-cluster solution.
\begin{proposition}
\label{prop:MCSol_Mix_FormalExpansion} Suppose $2a_{i,\mu}=a_{\mu}$
for all $\mu=1,\dots,M_{1}$, and $R_{2}(\mathbf{a}_{\mu})=0$ for
all $\mu=M_{1}+1,\dots,M$, $i=1,\dots,N_{\mu}$ where $M_1$ is the number of antipodal type clusters. The mixed pseudo
multi-cluster solutions of (\ref{eq:PhiDGL_general})\textendash (\ref{eq:KappaDGL_general})
with $\phi_{i,\mu}(t)=\Omega_{\mu}(\epsilon)t+s_{i,\mu}(t)+a_{i,\mu}$
possess the following first order asymptotic expansion in $\epsilon$
\begin{align*}
\phi_{i,\mu}^{(1)} & =\Omega^{(1)}_{\mu}t+a_{i,\mu}+\epsilon p_{i,\mu;1}(t),\\
\kappa^{(1)}_{ij,\mu\mu} & =-\sin(a_{i,\mu}-a_{j,\mu}+\beta),\\
\kappa^{(1)}_{ij,\mu\nu} & =\frac{\epsilon}{\Delta\Omega^{(1)}_{\mu\nu}}\cos(\Delta\Omega^{(1)}_{\mu\nu}t+a_{i,\mu}-a_{j,\nu}+\beta),
\end{align*}
with 
\begin{align*}
p_{i,\mu;1}(t)=p_{\mu;1} & =-\sum_{\overset{\nu=1}{\nu\ne\mu}}^{M_{1}}\frac{n_{\nu}}{4\left(\Delta\Omega^{(1)}_{\mu\nu}\right)^{2}}\cos(2\Delta\Omega^{(1)}_{\mu\nu}t+{a}_{\mu}-{a}_{\nu}+\alpha+\beta)
\end{align*}
for $\mu=1,\dots,M_{1}$,
\begin{align*}
	p_{i,\mu;1}(t) & =-\sum_{\overset{\nu=1}{\nu\ne\mu}}^{M}\frac{n_{\nu}}{4\left(\Delta\Omega^{(1)}_{\mu\nu}\right)^{2}}\cos(2\Delta\Omega^{(1)}_{\mu\nu}t+a_{i,\mu}-{a}_{\nu}+\alpha+\beta)
\end{align*}
for $\mu=M_{1}+1,\dots,M$, and the cluster frequencies $\Omega^{(1)}_{\mu}$ up to second order in
$\epsilon$ whenever the following system of equations can be solved
\begin{align}
\Omega^{(1)}_{\mu}=\left(\Omega_{\mu;0}-\epsilon\sum_{\overset{\nu=1}{\nu\ne\mu}}^{M}\frac{n_{\nu}}{2\Delta\Omega^{(1)}_{\mu\nu}}\sin(\alpha-\beta)\right)\label{eq:MCMix_Omega}
\end{align}
with
\begin{align*}
\Omega_{\mu;0} & =n_{\mu}\sin(\alpha)\sin(\beta) & \mu=1,\dots,M_{1}\\
\Omega_{\mu;0} & =\frac{n_{\mu}}{2}\cos(\alpha-\beta). & \mu=M_{1}+1,\dots,M
\end{align*}
 Here, $\mu=1,\dots,M,$ $i,j=1,\dots,N_{\mu}$ and $\Delta\Omega^{(1)}_{\mu\nu}:=\Omega^{(1)}_{\mu}-\Omega^{(1)}_{\nu}$.
\end{proposition}
As in the previous sections, we are going to show that equation~(\ref{eq:MCMix_Omega})
possesses solutions. For this, we consider the case of two clusters
$\phi_{i,\mu}$ $\left(\mu=1,2\right)$.
\begin{corollary}
\label{cor:2Cl_Mix_FormalExpan}Suppose $2a_{i,1}={a}_{1}$ for
all $i=1,\dots,N_{1}$ and $R(\mathbf{a}_{2})=0$. The mixed pseudo
multi-clusters of system (\ref{eq:PhiDGL_general})\textendash (\ref{eq:KappaDGL_general})
possess the following first order asymptotic expansion in $\epsilon$
\begin{align*}
\phi_{i,1}^{(1)} & =\Omega^{(1)}_{1}t+a_{i,1},\\
\phi_{i,2}^{(1)} & =\Omega^{(1)}_{2}t+a_{i,2}-\epsilon\frac{n_{1}}{4\left(\Delta\Omega^{(1)}_{12}\right)^{2}}\cos(2\Delta\Omega^{(1)}_{12}t+{a}_{1}-a_{i,2}-\alpha-\beta),\\
\kappa^{(1)}_{ij,\mu\mu} & =-\sin(a_{i,\mu}-a_{j,\mu}+\beta),\\
\kappa^{(1)}_{ij,\mu\nu} & =\frac{\epsilon}{\Delta\Omega^{(1)}_{\mu\nu}}\cos(\Delta\Omega^{(1)}_{\mu\nu}t+a_{i,\mu}-a_{j,\nu}+\beta),
\end{align*}
where
\[
\Omega^{(1)}_{\mu}=\left(\Omega_{\mu;0}-\epsilon\frac{n_{\nu}}{2\Delta\Omega^{(1)}_{\mu\nu}}\sin(\alpha-\beta)\right),
\]
\[
\Omega_{1;0}=n_{1}\sin(\alpha)\sin(\beta),
\]
\[
\Omega_{2;0}=\frac{n_{2}}{2}\cos(\alpha-\beta),
\]
\begin{align}
\begin{split}\left(\Delta\Omega^{(1)}_{12}\right)_{1,2} & =\frac{1}{2}\left[\left(n_{1}-\frac{1}{2}\right)\cos(\alpha-\beta)-\frac{n_{1}}{2}\cos(\alpha+\beta)\right]\\
 & \pm\frac{1}{2}\sqrt{\left[\left(n_{1}-\frac{1}{2}\right)\cos(\alpha-\beta)-\frac{n_{1}}{2}\cos(\alpha+\beta)\right]^{2}-2\epsilon\sin(\alpha-\beta)}
\end{split}
\label{eq:2ClusterMix_OmegaDiff}
\end{align}
for $\mu=1,2$, $\nu\ne\mu$ and $i,j=1,\dots,N_{\mu}$.
\end{corollary}
Illustration of the mixed 2-clusters is shown in Fig.~\ref{fig:2Cl_AntiPhaseRotWave}.
\begin{figure}
	\begin{center}
		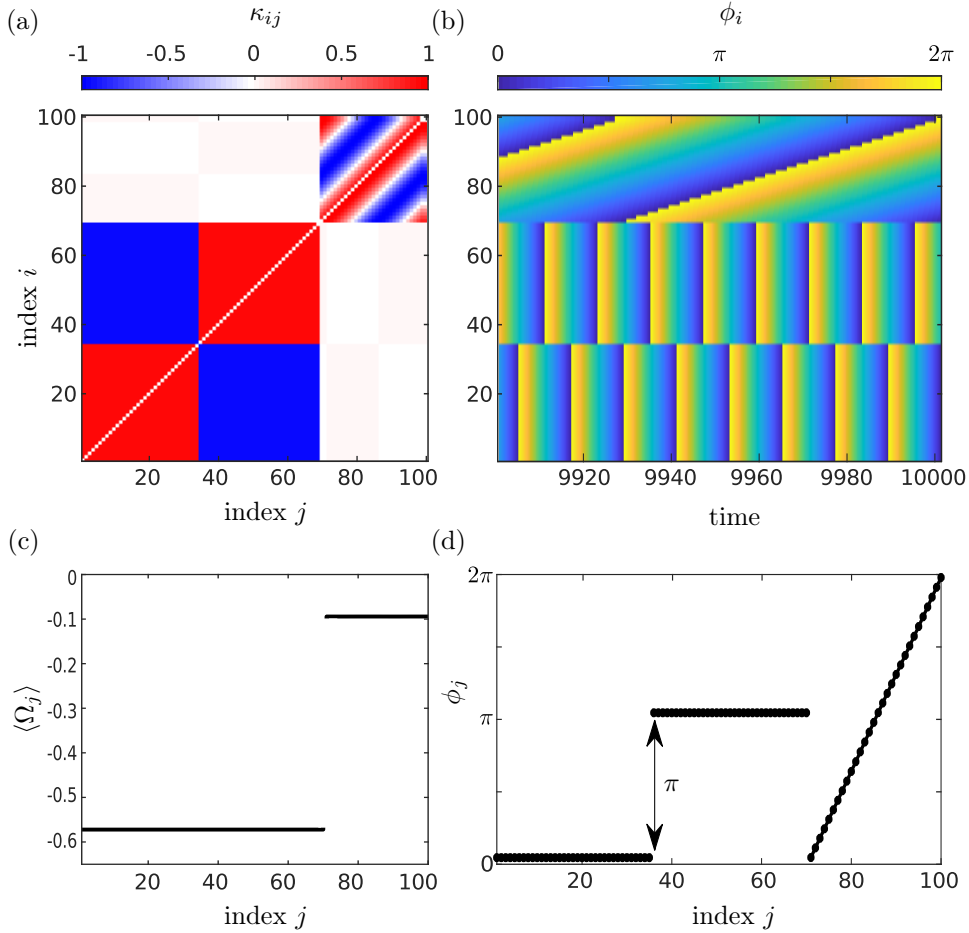
	\end{center}
\caption{$2$-Cluster solution of mixed type. (a) Coupling weights at $t=10000$
showing two clusters, (b) Distribution of the phases within each cluster,
space-time representation. (c) Average frequency of each oscillator,
(d) Oscillator phases $\phi_{i}$ for fixed time $t=10000$. Parameter
values: $\epsilon=0.01$, $\alpha=0.3\pi$, $\beta=-0.4\pi$, $N=100$. \label{fig:2Cl_AntiPhaseRotWave}}
\end{figure}
Moreover, we performed a Fourier analysis of the temporal behaviour
of the oscillators, see Fig.~\ref{fig:2Cl_AntiPhaseRotWave_Fourier_analysis}.
First, it can be observed that the oscillators representing the second
cluster ($i=80,100$) show the same evolution in time but with a phase
lag due to the spatial dependency described above. In order to show the agreement with the asymptotic expansion presented in Corollary~\ref{cor:2Cl_Mix_FormalExpan}, the analytic results are displayed with black dashed lines. Furthermore, the
power spectrum shows a prominent peak at $2\langle\Delta\Omega\rangle_{12}$ for
both oscillators of the second cluster and a flat curve for the representative
of the first cluster. These numerical results are in complete agreement
with the analytic findings. 
\begin{figure}
	\begin{center}
		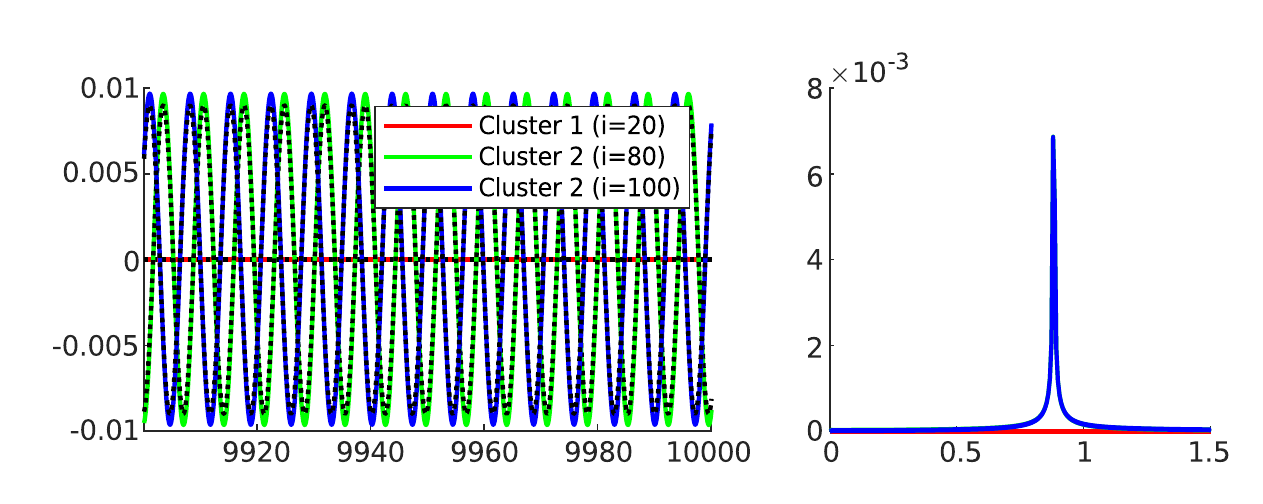
	\end{center}%
\caption{For mixed type 2-cluster solution from Fig.\ \ref{fig:2Cl_AntiPhaseRotWave},
panel (a) shows time series of an oscillator from one of the clusters
after subtracting the average linear growth $\phi_{\mu,i}(t)-\langle\Omega_{\mu}\rangle t$. The black dashed lines show the corresponding analytic results from the asymptotic expansion in Proposition~\ref{prop:MCSol_Mix_FormalExpansion}.
(b) Power spectrum of the time series given in (a).\label{fig:2Cl_AntiPhaseRotWave_Fourier_analysis}}
\end{figure}
Analogously, to the antipodal two-clusters, for any fixed set of parameters
and each $n_{1},$ equation~(\ref{eq:2ClusterMix_OmegaDiff}) can
lead to two multi-clusters of mixed type with two different frequency differences. Hence,
a large number of those clusters can coexist for the same parameter
values. Figure~\ref{fig:2Cl_Mix_OmegaSols} illustrates such
a coexistence, where we present the solutions to the equation~(\ref{eq:2ClusterMix_OmegaDiff}). Again, blue solid lines represent those solutions for which the asymptotic expansion led to an existing and stable two-cluster solutions of mixed type. Additionally, Figure~\ref{fig:2Cl_Mix_OmegaSols} shows the one-cluster solutions of splay and antipodal type (in both cases $\Delta\Omega_{12}=0$) together with their common regions of stability. As in the case of two-clusters of antipodal type,
the asymptotic expansion presented in Proposition~\ref{prop:MCSol_AP_FormalExpansion} turns into a formal expansion whenever $|\Delta\Omega|>\epsilon$. The interval $[-\epsilon,\epsilon]$
is therefore highlighted in Fig.~\ref{fig:2Cl_Mix_OmegaSols}.
\begin{figure}
	\begin{center}
		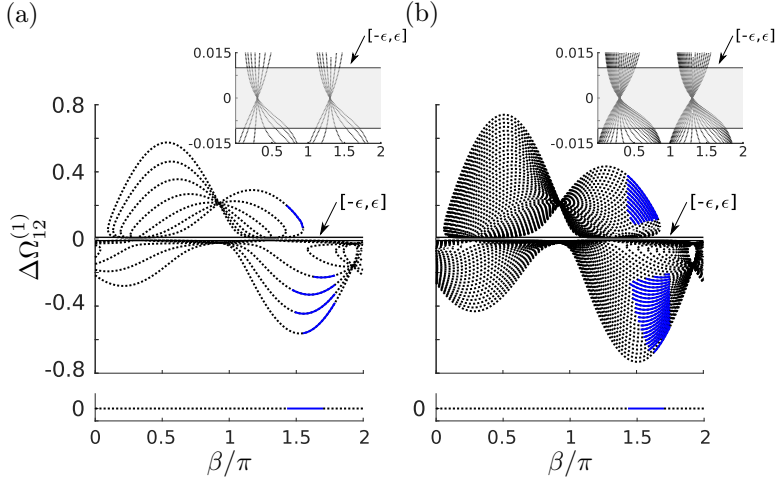
	\end{center}%
	\caption{\label{fig:2Cl_Mix_OmegaSols} Two-cluster solutions of mixed type (upper panels) and one-cluster solutions (lower panels) of either splay or antipodal type given by the asymptotic expansion in Corollary~\ref{cor:2Cl_Mix_FormalExpan} and Proposition~\ref{prop:1ClusterSolutions}, respectively. For this, the difference of the frequencies $\Delta\Omega^{(1)}_{12}$ is displayed corresponding to~(\ref{eq:2ClusterMix_OmegaDiff}) and~(\ref{eq:OmegaPL}). The dotted lines (black) indicate unstable solutions while the solid	lines (blue) indicate stable solutions. Here, every second solution is	plotted for the sake of visibility. The insets show a blow-up of the interval $[-\epsilon,\epsilon]$. Parameter values: (a) $N=20$, $\epsilon=0.01$; (b) $N=50$, $\epsilon=0.01$; $\alpha=0.3\pi$ is fixed for all panels.}
\end{figure}

\section{Stability of cluster solutions\label{sec:StabilityAnalysis}}

In sections~\ref{sec:One-cluster}\textendash \ref{sec:Multi-cluster-states}
the existence of one-cluster as well as (pseudo) multi-cluster solutions
were discussed. This section is devoted to the analysis of their stability.
First, the stability of the one-cluster solutions is analyzed and an
analytic result for all rotating-wave one-cluster solutions which are
described in section~\ref{subsec:stab-One-cluster-rotating-wave} is
presented. The findings are discussed with respect to all one-cluster
solutions found in Proposition~\ref{prop:1ClusterSolutions}. Subsequently,
we use numerical simulations in order to analyze the stability for
two-cluster solutions. 

\subsection{One-cluster solutions\label{subsec:stab-One-cluster-rotating-wave}}

In order to study the local stability of one-cluster solutions described
in section~\ref{sec:One-cluster}, we linearize the system  (\ref{eq:PhiDGL_general})\textendash (\ref{eq:KappaDGL_general})
around the solutions (\ref{eq:pPL})\textendash (\ref{eq:kPL}). We obtain
the following linearized system 
\begin{align}
\frac{d}{dt}\delta\phi_{i} & =\frac{1}{N}\sum_{m=0}^{N-1}\sin(a_{i}-a_{i+m}+\beta)\cos(a_{i}-a_{i+m}+\alpha)\left(\delta\phi_{i}-\delta\phi_{i+m}\right)\label{eq:Linearized_OneCl_phi}\\
 & -\frac{1}{N}\sum_{m=0}^{N-1}\sin(a_{i}-a_{i+m}+\alpha)\delta\kappa_{i(i+m)},\nonumber \\
\frac{d}{dt}\delta\kappa_{i(i+m)} & =-\epsilon\left(\delta\kappa_{i(i+m)}+\cos(a_{i}-a_{i+m}+\beta)\left(\delta\phi_{i}-\delta\phi_{i+m}\right)\right),\label{eq:Linearized_OneCl_kappa}
\end{align}
where we have introduced the new label $m:=j-i$ and the convention $i+m=(i+m)\mod N$
for convenience. Throughout this paragraph we will make use of Schur's complement~\cite{BOY04} in order to simplify characteristic equations. More precisely, any $m\times m$ matrix $M$ in the $2\times2$ block form can be written as
\begin{align}\label{eq:SchurComplement}
M & =\begin{pmatrix}A & B\\
C & D
\end{pmatrix}
=\begin{pmatrix}\mathbb{I}_{p} & BD^{-1}\\
	0 & \mathbb{I}_{q}
\end{pmatrix}\begin{pmatrix}A-BD^{-1}C & 0\\
	0 & D
\end{pmatrix}\begin{pmatrix}\mathbb{I}_{p} & 0\\
	D^{-1}C & \mathbb{I}_{q}
\end{pmatrix}
\end{align}
where $A$ is a $p\times p$ matrix and $D$ is an invertible $q\times q$
matrix. The matrix $A-BD^{-1}C$ is called Schur's
complement. A simple formula for the determinant of $M$ can be derived
with the decomposition~\eqref{eq:SchurComplement}
\begin{align*}
\det(M) & =\det(A-BD^{-1}C)\cdot\det(D).
\end{align*}
This result is important for the subsequent stability analysis. Note that in the following an overline indicates the complex conjugate.
\begin{lemma}
\label{lem:LinearizedOneCluster_BlockForm}Suppose $\mathbf{a}_{k}=(0,2\pi k/N,\dots,(N-1)2\pi k/N)^{T}$
with $k\in\{0,\dots,N-1\}$ and the linear system around the one-cluster
solution $\bm{\phi}=\Omega t\cdot(1,\dots,1)^{T}+\mathbf{a}_{k}$
is given by (\ref{eq:Linearized_OneCl_phi})\textendash (\ref{eq:Linearized_OneCl_kappa}).
Then there exist new coordinates $(\delta\psi, \delta\zeta)$ such that the linearized system can be decomposed into $N$ linear differential equations of the form
\begin{align}\label{eq:LinearizedOneCl_BlockForm}
\left(\begin{matrix}\delta\dot{\psi}_{l}\\
\delta\dot{\zeta}_{l0}\\
\vdots\\
\delta\dot{\zeta}_{l(N-1)}
\end{matrix}\right)
=C_l\left(\begin{matrix}\delta{\psi}_{l}\\
\delta{\zeta}_{l0}\\
\vdots\\
\delta{\zeta}_{l(N-1)}
\end{matrix}\right) \quad l=0,\dots,N-1
\end{align}
with
\begin{align*}
C_{l} & :=\left(\begin{matrix}\hat{{\lambda}}_{l} & \begin{matrix}b\end{matrix}\\
\begin{matrix}c_{l}\end{matrix} & -\epsilon\mathbb{I}_{N}
\end{matrix}\right),
\end{align*}
where, $\mathbb{I}_{N}$ is the $N$-dimensional identity matrix and
\begin{align}
\hat{{\lambda}}_{l} & =\frac{1}{2}\left((Z_{1}(\mathbf{a}_{l})-1)\sin(\alpha-\beta)-\Im(Z_{2}(\mathbf{a}_{k}))\cos(\alpha+\beta)+\Re(Z_{2}(\mathbf{a}_{k}))\sin(\alpha+\beta)\right)\label{eq:LinearizedOneCl_lambdabar}\\
 & +\frac{1}{4}\left(\overline{Z}_{1}(\mathbf{a}_{2k-l})\mathrm{i}e^{\mathrm{i}(\alpha+\beta)}-Z_{1}(\mathbf{a}_{2k+l})\mathrm{i}e^{-\mathrm{i}(\alpha+\beta)}\right).\nonumber \\
b & =\frac{1}{N}\left(\sin(-\alpha),\dots,\sin((N-1)k\frac{2\pi}{N}-\alpha)\right),\label{eq:LinearizedOneCl_b}\\
c_{l} & =\left(0,\cos(k\frac{2\pi}{N}-\beta)\left(1-e^{\mathrm{i}l\frac{2\pi}{N}}\right),\dots,\cos((N-1)k\frac{2\pi}{N}-\beta)\left(1-e^{\mathrm{i}l(N-1)\frac{2\pi}{N}}\right)\right)^{T}\label{eq:LinearizedOneCl_c}
\end{align}
with any $j\in\{1,\dots,N\}$.
\end{lemma}
\begin{proof}
Due to the cyclic structure in the equations (\ref{eq:Linearized_OneCl_phi})
and (\ref{eq:Linearized_OneCl_kappa}) it is possible to decouple
them using a discrete Fourier ansatz~\cite{PER10c}
\begin{align*}
\delta\phi_{j} & =\sum_{l=0}^{N-1}e^{\mathrm{i}lj\frac{2\pi}{N}}\delta\psi_{l},\\
\delta\kappa_{j(j+m)} & =\sum\limits _{l=0}^{N-1}e^{\mathrm{i}lj\frac{2\pi}{N}}\delta\zeta_{lm}.
\end{align*}
 Taking this Fourier ansatz and plugging it into the equations (\ref{eq:Linearized_OneCl_phi})
and (\ref{eq:Linearized_OneCl_kappa}) we get
\begin{align*}
\sum_{l=0}^{N-1}e^{\mathrm{i}lj\frac{2\pi}{N}}\dot{\delta\psi}_{l}= & \frac{1}{N}\sum_{m=0}^{N-1}\sin(-mk\frac{2\pi}{N}+\beta)\cos(-mk\frac{2\pi}{N}+\alpha)\sum_{l=0}^{N-1}e^{\mathrm{i}lj\frac{2\pi}{N}}\left(1-e^{\mathrm{i}lm\frac{2\pi}{N}}\right)\delta\psi_{l}\\
 & -\frac{1}{N}\sum_{m=0}^{N-1}\sin(-mk\frac{2\pi}{N}+\alpha)\sum\limits _{l=0}^{N-1}e^{\mathrm{i}lj\frac{2\pi}{N}}\delta\zeta_{lm},\\
\sum\limits _{l=0}^{N-1}e^{\mathrm{i}lj\frac{2\pi}{N}}\dot{\delta\zeta}_{lm} & =-\epsilon \sum\limits _{l=0}^{N-1} \left(e^{\mathrm{i}lj\frac{2\pi}{N}}\delta\zeta_{lm}+\cos(-mk\frac{2\pi}{N}+\beta)e^{\mathrm{i}lj\frac{2\pi}{N}}\left(1-e^{\mathrm{i}lm\frac{2\pi}{N}}\right)\delta\psi_{l}\right).
\end{align*}
After making use of well known trigonometric identities and using
the order parameters defined in (\ref{eq:Nth_OrderParameter}) we
find
\begin{align*}
\hat{{\lambda}}_{l} & =\frac{1}{2N}\sum_{m=0}^{N-1}\left(\sin(-\frac{4\pi}{N}mk+\alpha+\beta)-\sin(\alpha-\beta)\right)\left(1-\cos(lm\frac{2\pi}{N})-\mathrm{i}\sin(lm\frac{2\pi}{N})\right)\\
 & =\frac{1}{2}\left((Z_{1}(\mathbf{a}_{l})-1)\sin(\alpha-\beta)-\Im(Z_{2}(\mathbf{a}_{k}))\cos(\alpha+\beta)+\Re(Z_{2}(\mathbf{a}_{k}))\sin(\alpha+\beta)\right)\\
 & +\frac{1}{4}\left(\overline{Z}_{1}(\mathbf{a}_{2k-l})\mathrm{i}e^{\mathrm{i}(\alpha+\beta)}-Z_{1}(\mathbf{a}_{2k+l})\mathrm{i}e^{-\mathrm{i}(\alpha+\beta)}\right).
\end{align*}
The row and the column vectors $b_{l}$ and $c_{l}$ can directly
be read of from the transformed equation above.
\end{proof}
Note that the values $\hat{{\lambda}}_{l}$ are exactly the eigenvalues
for the case where no interaction between the oscillators and their
coupling are assumed or the dynamics of the coupling weights are left
constant. One might expect that due to the slow-fast dynamics of the
system (\ref{eq:PhiDGL_general})\textendash ~(\ref{eq:KappaDGL_general})
a small perturbation in the coupling weights could be neglected for
the analysis of stability \cite{AOK11}. In contrast, we show that
the local dynamics of the system around the one-cluster solution depends
on the interplay between phases and couplings.
\begin{proposition}
\label{prop:LinerizedOneCluster_RW_Spectrum}Suppose $\mathbf{a}_{k}=(0,\frac{2\pi}{N}k,\dots,(N-1)\frac{2\pi}{N}k)^{T}$
and the linear system around the one-cluster solution $\bm{\phi}=\Omega t\cdot(1,\dots,1)^{T}+\mathbf{a}_{k}$
is given by (\ref{eq:Linearized_OneCl_phi})\textendash (\ref{eq:Linearized_OneCl_kappa}).
Then the Jacobian $J$ of this linearized system possesses the following spectrum
\begin{align*}
\sigma(J) & =\left\{ -\epsilon,(\lambda_{l;1,2})_{l=0,\dots,N-1}\right\} 
\end{align*}
 with
\begin{align}
\lambda_{l;1,2} & =\frac{\hat{{\lambda}}_{l}-\epsilon}{2}\pm\frac{1}{2}\sqrt{(\hat{{\lambda}}_{l}+\epsilon)^{2}+4\epsilon\left(b\cdot c\right)_{l}}\label{eq:NonTrivialEVs_RW}
\end{align}
 and $\hat{{\lambda}}_{l}$, $b_{l}$, $c_{l}$ as defined in (\ref{eq:LinearizedOneCl_lambdabar}),
(\ref{eq:LinearizedOneCl_b}) and (\ref{eq:LinearizedOneCl_c}).
\end{proposition}
\begin{proof}
Using Lemma~\ref{lem:LinearizedOneCluster_BlockForm} we decompose
the linear system (\ref{eq:Linearized_OneCl_phi})\textendash (\ref{eq:Linearized_OneCl_kappa})
into the $N$ blocks~(\ref{eq:LinearizedOneCl_BlockForm}).
Consider now the characteristic polynomial for the $(N+1)\times(N+1)$
matrix $C_{l}$ and assume that $\lambda_{l}\ne-\epsilon$ then by~\eqref{eq:SchurComplement} we obtain
\begin{align*}
\det(\lambda_{l}\mathbb{I}_{N+1}-C_{l})=\det\left(\begin{matrix}\lambda_{l}-\hat{{\lambda}}_{l} & -b_{l}\\
-\epsilon c_{l} & (\epsilon+\lambda_{l})\mathbb{I}_{N}
\end{matrix}\right)\\
=(\epsilon+\lambda_{l})^{N-1}\left((\epsilon+\lambda_{l})(\lambda_{l}-\hat{{\lambda}}_{l})-\epsilon\left(b\cdot c\right)_{l}\right)=0.
\end{align*}
Thus for each $l\in{0,\dots,N-1}$ there are $N-1$ eigenvalues $\lambda_{l}=-\epsilon$.
For the two remaining eigenvalues we have to solve the quadratic equation
\begin{align}
\lambda_{l}^{2}+(\epsilon-\hat{{\lambda}}_{l})\lambda_{l} & -\epsilon\hat{{\lambda}}_l-\epsilon\left(b\cdot c\right)_{l}=0.\label{eq:LinearizedOneCl_EV}
\end{align}
\end{proof}
In the case of no weight dynamics or no coupling between the oscillators
and the weights the eigenvalues would read $\lambda_{l;1}=\hat{{\lambda}}_{l}$
and $\lambda_{l;2}=-\epsilon$. Therefore, the spectrum would look
like $\sigma_{c}=\{-\epsilon,(\hat{{\lambda}}_{l})_{l=0,\dots,N-1}\}$
with $(N-1)N$-fold multiplicity for the eigenvalue $-\epsilon.$
In contrast to that, we get in general $2N$ eigenvalues that are
different from $-\epsilon$ which stem from the interplay of phases
and coupling weights. We should further mention that $\hat{{\lambda}}_{l}(\alpha+\frac{\pi}{2},\beta-\frac{\pi}{2})=(b\cdot c)_{l}(\alpha,\beta)$.
With this we write equation (\ref{eq:LinearizedOneCl_EV}) as
\begin{align*}
\lambda_{l}^{2}(\alpha,\beta)+\left(\epsilon-\hat{{\lambda}}_{l}(\alpha,\beta)\right)\lambda_{l}-\epsilon\left(\hat{{\lambda}}_{l}(\alpha,\beta)+\hat{{\lambda}}_{l}(\alpha-\frac{\pi}{2},\beta+\frac{\pi}{2})\right) & =0.
\end{align*}
The following corollary summarizes the results on the spectrum of
the linearized system~(\ref{eq:Linearized_OneCl_phi})\textendash (\ref{eq:Linearized_OneCl_kappa}).
\begin{corollary}
\label{cor:LinerizedOneCluster_RW_Spectrum}Suppose we have $\mathbf{a}_{k}=(0,\frac{2\pi}{N}k,\dots,(N-1)\frac{2\pi}{N}k)^{T}$
and the linear system (\ref{eq:Linearized_OneCl_phi})\textendash (\ref{eq:Linearized_OneCl_kappa})
then
\begin{enumerate}
\item (in-phase and anti-phase synchrony) if $k=0$ or $k=N/2$, the spectrum
is given by
\[
\sigma(C)=\left\{ \left(0\right)_{\text{1}},\left(-\epsilon\right)_{(N-1)N+1},\left(\lambda_{1}\right)_{N-1},\left(\lambda_{2}\right)_{N-1}\right\} 
\]
where $\lambda_{1}$ and $\lambda_{2}$ solve $\lambda^{2}+\left(\epsilon-\cos(\alpha)\sin(\beta)\right)\lambda-\epsilon\sin(\alpha+\beta)=0,$
\item (incoherent rotating-wave) if $k\ne0,N/2,N/4,3N/4,$ the spectrum
is 
\begin{multline*}
\sigma(C) =\left\{ \left(0\right)_{N-2},\left(-\epsilon\right)_{(N-1)N+1},\left(-\frac{\sin(\alpha-\beta)}{2}-\epsilon\right)_{N-3},\right.\\
\left.\left(\vartheta_{1}\right)_{1},\left(\vartheta_{2}\right)_{1},\left(\overline{\vartheta}_{1}\right)_{1},\left(\overline{\vartheta}_{2}\right)_{1}\right\} 
\end{multline*}
where $\vartheta_{1}$ and $\vartheta_{2}$ solve $\vartheta^{2}+\left(\epsilon+\frac{1}{2}\sin(\alpha-\beta)-\frac{1}{4}\mathrm{i}e^{\mathrm{i}(\alpha+\beta)}\right)\vartheta-\frac{\epsilon}{2}\mathrm{i}e^{\mathrm{i}(\alpha+\beta)}=0,$
\item (4-rotating-wave solution) if $k=N/4,3N/4$, the spectrum is 
\[
\sigma(C)=\left\{ \left(0\right)_{N-1},\left(-\epsilon\right)_{(N-1)N+1},\left(-\frac{\sin(\alpha-\beta)}{2}-\epsilon\right)_{N-2},\left(\lambda_{1}\right)_{1},\left(\lambda_{2}\right)_{1}\right\} 
\]
where $\lambda_{1}$ and $\lambda_{2}$ solve $\lambda^{2}+\left(\epsilon+\sin(\alpha)\cos(\beta)\right)\lambda+\epsilon\sin(\alpha+\beta)=0.$
\end{enumerate}
Here, the multiplicities for each eigenvalue are given as lower case
labels.
\end{corollary}
As we can see from this corollary there exists always at least one
zero eigenvalue. This is due to the phase-shift symmetry of
(\ref{eq:PhiDGL_general})\textendash (\ref{eq:KappaDGL_general})
we already discussed in section~\ref{sec:Introduction}. The additional
zero eigenvalues for the wave numbers $k\ne0,N/2$ can be explained
with our findings from Proposition~\ref{prop:1ClusterSolutions}
and Corollary~\ref{cor:1ClusterRWSol}. These linear rotating-wave
solutions belong to a $N-2$ dimensional family of solutions characterized
by $R_{2}(\mathbf{a})=0$. Thus, around any point of this family the
linear equation
\begin{align}
\sum_{j=1}^{N} & e^{i2a_{j}}\delta\phi_{j}=0\label{eq:RotWave_Manifold_LocalCondition}
\end{align}
holds for a certain choice of coordinates $\left\{ \delta\phi,\delta\kappa_{ij}\right\} $,
and hence there are two linearly independent equations for the infinitesimal
perturbations $\delta\phi_{i}$. This explains the appearance of $N-2$
zero eigenvalues. They correspond to the variation along the manifold
of solution. In the special case of $k=N/4,3N/4$ the two algebraic
equations~(\ref{eq:RotWave_Manifold_LocalCondition}) are linear
dependent and we are thus left with only one linear equation which
increases the multiplicity of the zero eigenvalue by one.
The results of Corollary~\ref{cor:LinerizedOneCluster_RW_Spectrum}
are presented in Fig.~\ref{fig:Stability_RW_AP_1Cl}(a\textendash c) and compared with numerical simulations.
The numerical results are obtained by numerical integration of system~(\ref{eq:PhiDGL_general})\textendash (\ref{eq:KappaDGL_general}) with $N=20$. The initial conditions for each simulation are set to the one-cluster
solution given in Proposition~\ref{prop:1ClusterSolutions} with a small perturbation
added to each dynamical variable and randomly chosen from the interval
$[-0.01,0.01]$. The numerical integration is stopped after $t=5000$ time steps. The relative coordinates $\Theta_i:=\phi_i-\phi_1$ for $i=1,\dots,N$ are introduced in order to compare the initial phase configuration with the distribution of the phases after numerical integration. A one-cluster is said to be stable if $\bm{\Theta}$ after numerical integration is closer to the theoretical one-cluster state than $\bm{\Theta}$ of the initially perturbed phase distribution. Closeness is measured by the Euclidean distance. Otherwise, the one-cluster solution is considered as unstable. The parameter regions in the $(\alpha,\beta)$ plane for stable one-cluster
solutions are coloured blue while the regions for unstable one-cluster solutions
are coloured yellow. The black dashed lines correspond to the borders of stability determined with the results in  Corollary~\ref{cor:LinerizedOneCluster_RW_Spectrum}. In particular, a state is asymptotically stable if $\Re(\lambda)<0$ for  all $\lambda\in\sigma(C)$ except the zero eigenvalues related to the perturbations along the solution families.
In all three cases the numerical and analytic results agree very well.
\begin{figure}
	\begin{center}
		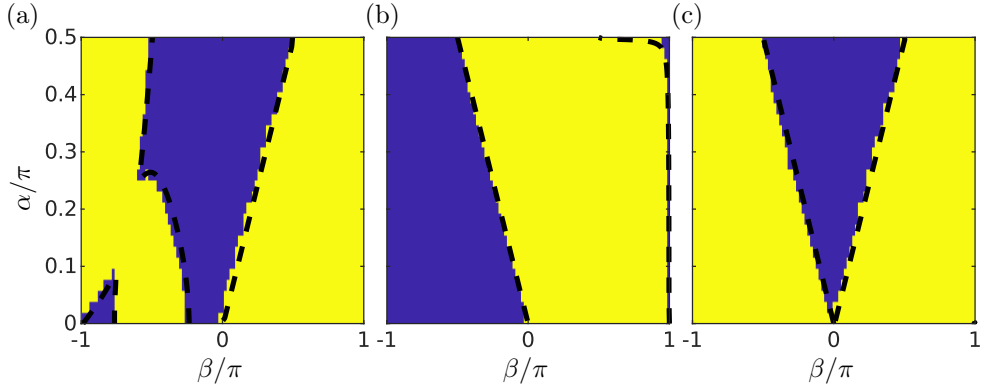
	\end{center}%
\caption{\label{fig:Stability_RW_AP_1Cl}Stability diagrams for rotating-wave
clusters depending on the parameters $\alpha$ and $\beta$ are shown. The regions are coloured according to numerical simulation. Blue regions correspond to stable solutions while yellow regions correspond to unstable solutions. The black dashed lines show to the borders of stability determined by Corollary~\ref{cor:LinerizedOneCluster_RW_Spectrum}. Parameter $\epsilon=0.01$ is fixed fo all simulations. (a) $k=1$, (b) $k=N/2$, (c) $k=N/4$} 
\end{figure}

In addition to the analysis of rotating-wave solutions, we investigate
the stability for the splay solutions characterized by $R(\mathbf{a})=0$
and the antipodal solutions characterized by $R(\mathbf{a})=1$.
For this, the stability is calculated by taking the solutions displayed
in Fig.~\ref{fig:1Cl_Illustration}(a\textendash b), plugging them into the Jacobian matrix given by the linearized equations~\ref{eq:Linearized_OneCl_phi}\textendash \ref{eq:Linearized_OneCl_kappa}
and determining the eigenvalues of the Jacobian numerically. The results of this procedure
are shown in Fig.~\ref{fig:Stability_type1_type2_1Cl} together with the borders of stability calculated with Corollary~\ref{cor:LinerizedOneCluster_RW_Spectrum}. In comparison
with Fig.~\ref{fig:Stability_RW_AP_1Cl}, the analysis yields the
same stability regions which are found for the rotating-wave solutions.
The numerical findings indicate that the stability for all splay and antipodal solutions coincide with the stability
of the rotating-waves.
\begin{figure}
	\begin{center}
		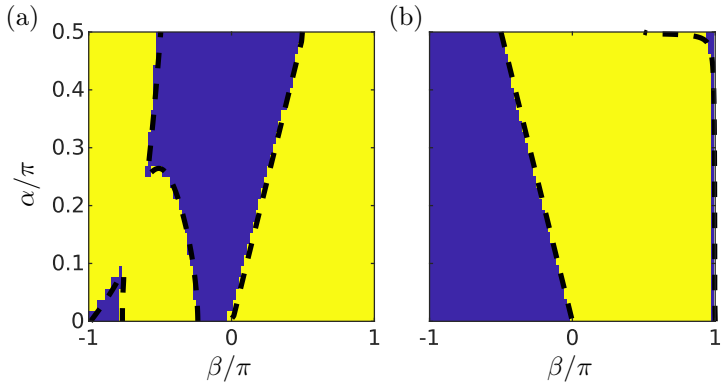
	\end{center}%
\caption{\label{fig:Stability_type1_type2_1Cl}Stability diagram for splay and antipodal one-cluster solutions depending on the parameters
$\alpha$ and $\beta$ are shown. The regions are coloured according numerical eigenvalues of the Jacobian~\ref{eq:Linearized_OneCl_phi}\textendash \ref{eq:Linearized_OneCl_kappa}. Blue areas correspond to
stable while yellow areas correspond to unstable regions. Parameter $\epsilon=0.01$ is fixed in all simulations. (a) Splay solution as
in Fig.~\ref{fig:1Cl_Illustration}(a), (b) Anti-phase solution as in
Fig.~\ref{fig:1Cl_Illustration}(b).}
\end{figure}

\subsection{Stability of multi-cluster solutions}

In section~\ref{subsec:MC_RotWave} we discussed multi-cluster solutions
of splay type and showed under which condition they exist.
The solutions for two-cluster solutions of splay type and their
stability are presented in Fig.~\ref{fig:2Cl_RotWave_OmegaSols}.
In Fig.~\ref{fig:2Cl_RotWave_OmegaSols}(b) the solution for the
case of $50$ oscillators is shown. The solid lines (blue) correspond
to solutions that are stable. It can be seen that whenever a 2-cluster
solution is stable the one-cluster solution (with $\Delta\Omega_{12}=0$)
is also stable. A more detailed validation of this statement is presented
in Fig.~\ref{fig:2Cl_vs_1Cl_RW_stability}, where we show the stability
regions of both one- and two-cluster solutions in the $(\alpha,\beta)$ plane. The stability for each type of cluster solution is determined numerically. The numerical approach was already introduced in section~\ref{subsec:stab-One-cluster-rotating-wave}. For the two-cluster solutions the norm for the phase configuration is calculated in the relative coordinates given by $\Theta_{i,\mu}=\phi_{i,\mu}-\phi_{1,\mu}$ with $\mu=1,2$. Additionally, we calculated the maximal value of all inter-cluster connections and compared it to the theoretical maximum given by $\rho_{12}$. If after numerical integration the maximal inter-cluster coupling is bigger than $\rho_{12}+0.01$, the two-cluster is considered as unstable.
Here, region where the both types of solutions are stable are colored
in dark blue. Regions of only stable one-cluster solutions are colored
in light blue. Since two-cluster solutions do not exist for certain values
of $\alpha$ and $\beta$, we can find a light blue stripe in the
middle of Fig.~\ref{fig:2Cl_vs_1Cl_RW_stability}. Further, we have
not found any configuration of $\alpha$ and $\beta$ for which two-cluster
solutions are stable and one-cluster are not. This supports the claim
that the stability of a one-cluster solution is necessary condition for
the stability of a two-cluster solution. This can be explained by the
fact that for the stability of the multi-clusters, it is necessary
that its one-cluster components are each stable with respect to the
perturbations that disturb the structure of just one cluster (see
similarly in \cite{LUE12a}). A more rigorous formulation of this
issue is beyond the scope of this paper.

Figure~\ref{fig:2Cl_RotWave_OmegaSols}(b) further provides us with
information about the stability of two-cluster solutions depending on
the ratio between cluster sizes. First, due to~(\ref{eq:2Cluster_OmegaDiff})
there exist two branches of two-cluster solution of splay
type. Only solutions with higher frequency difference are stable which
can be seen in the inset of Fig.~\ref{fig:2Cl_RotWave_OmegaSols}(b).
For an increasing number of oscillators in the second cluster of relative size $n_{2}=1-n_{1}$
the stability changes. Above a certain value of $n_{2}$ both branches
are unstable. This observation explains why only multi-cluster solutions with
unequal as well as hierarchical cluster sizes were found in simulations,
see Fig.~(\ref{fig:3Cl_RotWave}) and \cite{KAS17}.
\begin{figure}
	\begin{center}
		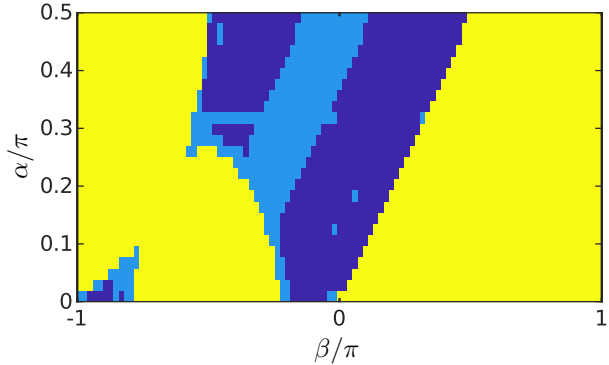
	\end{center}%
\caption{\label{fig:2Cl_vs_1Cl_RW_stability}Stability diagram for the one-cluster
and two-cluster solution of the splay type depending on the
parameters $\alpha$ and $\beta$. Yellow region corresponds to the
instability of both solutions, dark blue to the stability of
both solutions, and light-blue to the stability of only the one-cluster solution. Parameter $\epsilon=0.01$ is fixed in all simulations.}
\end{figure}

\section{Conclusions\label{sec:conclusion}}

In summary, this work provides a wide-ranging analysis of multi-cluster
solutions in networks of adaptively coupled phase oscillators. It covers questions on the existence,
the explicit form as well as the stability of such solutions. It is
well known that under certain conditions systems of coupled oscillators
can be simplified to systems of coupled phase oscillators. In these
cases, the model of phase oscillators captures the underlying dynamics
of complex dynamical systems. In this sense, the considered model
may have multiple applications. In particular, the model which
was in the focus of this work was inspired by dynamical neuronal networks
with adaptive plasticity. Therefore, it may help to understand the
fundamental mechanisms of pattern formation in neuronal systems.

Here, we have focused on multi-cluster solutions, which are composed
of several one-clusters with distinct frequencies. The one-cluster
solutions are shown to be of the following three types: splay type, antipodal type, and a new third type, named double antipodal, see Proposition~\ref{prop:1ClusterSolutions}.
Moreover, it was shown that all one-cluster solutions of splay
type form an $N-2$ dimensional family and thus give rise to infinitely
many solutions the system can achieve. With this, we have generalized and extended the results reported in
\cite{AOK11,GUS15a,NEK16}. 

While the one-cluster solutions are relative equilibria of our system
due to the phase-shift symmetry, the multi-cluster solutions contain
components with different frequencies, and, hence, they cannot be
reduced to an equilibrium by transforming into another co-rotating
frame. As a result, the study of multi-clusters is more involved.
However, to our surprise, we have still been able to find an explicit
form of multi-clusters with the components of the splay type.
Remarkably, in addition to its ring-like spatial structure that dynamically
emerges, the network behaves in such a case (quasi-)periodically in
time such that the whole solution can be interpreted as a spatial-temporal
wave. 

The analysis of multi-cluster solutions of antipodal type is more subtle due to the modulation of the frequency. More specifically, we look at multi-clusters with bounded frequency modulation. For these types of multi-clusters, we derive an asymptotic expansion in the parameter $\epsilon$ that gives explicit existence conditions. 

In addition, we have shown the existence of mixed multi-clusters, which
consist of clusters of splay type and clusters of antipodal type. For the mixed multi-clusters, the temporal behavior
within one cluster has been shown to be slightly non-identical, namely,
the oscillators possess the same averaged frequency, but they still
can have a bounded quasi-periodically modulated phase difference. 

We have been able to study the stability of multi-clusters analytically to some extent,
and otherwise numerically. The main messages from this analysis are
as follows: there is a high degree of coexistence of stable multi-clusters
that can be reached from different initial conditions; in particular,
a certain amount of imbalance in the number of oscillators within
the clusters is needed to achieve stability. This explains the
appearance of only hierarchical structures in numerical simulations.

Moreover, the findings on multi-cluster solutions as they are reported in this article are in very good agreement with previous results on adaptive neural networks~\cite{POP15}. Here, stable multi-cluster solutions of coherently spiking neurons with weak but time-dependent inter-cluster coupling are reported. With this work we shed some light on these generic time-dependent network patterns.

\newpage
\appendix

\section{Proofs of propositions from sections \ref{sec:One-cluster}\textendash \ref{sec:Multi-cluster-states}\label{sec:Proofs}}

\subsection{One-cluster solution}

Here we provide a proof of Proposition \ref{prop:1ClusterSolutions}.
We first need a preliminary lemma. 
\begin{lemma}
\label{lem:R2equal1} For a phase-locked solution $\bm{\phi}(t)$, $R_{2}(\bm{\phi}(t))=1$
for all $t$ if and only if $\bm{\phi}(t)$ is either an in-phase
or an anti-phase synchronous solution.
\end{lemma}
\begin{proof}
As follows from table~\ref{tab:Nth_OrderParameter_DiffStates}, $R_{2}(\mathbf{a})=\frac{1}{N}\left|\sum_{j=1}^{N}e^{\mathrm{i}2a_{j}}\right|=1$
for all in-phase and anti-phase solutions. Let us show the opposite.
If $R_{2}(\mathbf{a})=1$, then $e^{\mathrm{i}2a_{j}}=e^{\mathrm{i}2a_{1}}=1$
for all $j$, since $\left|\sum_{j=1}^{N}e^{\mathrm{i}2a_{j}}\right|\le\sum_{j=1}^{N}\left|e^{\mathrm{i}2a_{j}}\right|=N$.
Hence, $a_{j}\in\left\{ 0,\pi\right\} $. The latter means that the
phase-locked solution is either in-phase, if all $a_{j}$ have the
same values, or anti-phase otherwise.
\end{proof}
Now we present the proof of Proposition \ref{prop:1ClusterSolutions}.
\begin{proof}
Substituting (\ref{eq:pPL})\textendash (\ref{eq:kPL}) into (\ref{eq:PhiDGL_general})\textendash (\ref{eq:KappaDGL_general})
we obtain $\dot{\kappa}_{ij}=0$ and 
\begin{align}
\dot{\phi}_{i}(t)=\Omega & =\frac{1}{N}\sum_{j=1}^{N}\sin(a_{i}-a_{j}+\beta)\sin(a_{i}-a_{j}+\alpha)\nonumber \\
 & =\frac{1}{2}\cos(\alpha-\beta)-\frac{1}{2}\Re\left(e^{-\mathrm{i}(2a_{i}+\alpha+\beta)}Z_{2}(a)\right).\label{eq:Omega_General}
\end{align}
Therefore solution (\ref{eq:pPL})\textendash (\ref{eq:kPL}) are
solutions if and only if the expression on the right hand side of
the equation~(\ref{eq:Omega_General}) is independent of the oscillators
index $i=1,\dots,N$. In particular, for any choice of $a_{i}$ the
complex second order parameter is either zero or can be written as
$Z_{2}(\mathbf{a})=R_{2}(\mathbf{a})e^{\mathrm{-i}\gamma}$. Thus,
according to (\ref{eq:Omega_General}) $a_{i}$ has to be such that
$R_{2}(\bm{a})=0$ or $\cos(2a_{i}+\alpha+\beta+\gamma)$ is independent
of $i$. For any $\alpha,\beta$ and $\gamma$ the latter requirement
is equivalent with $2a_{i}\in\left\{ 0,-2(\alpha+\beta+\gamma)\right\} $.
Here, we made use of the phase-shift symmetry
of (\ref{eq:PhiDGL_general})\textendash (\ref{eq:KappaDGL_general})
by setting $2a_{1}=0$. Due to the definition of the complex order
parameter the value for $\gamma$ depends on the choice of the phase
lags $a_{i}$. Assuming that one fraction $q_{1}=Q_{1}/N$ of the
oscillators have $2a_{i}=0$ with $Q_1\in\{1,\dots,N-1\}$ and the other fraction of oscillators
$q_{2}=1-q_1$ have $2a_{i}=-2(\alpha+\beta+\gamma)$ one obtains
\begin{align*}
q_{1}+q_{2}e^{-\mathrm{i}2(\alpha+\beta)}e^{-\mathrm{i}2\gamma} & =R_{2}(\mathbf{a})e^{-\mathrm{i}\gamma}.
\end{align*}
which is equivalent to the equations
\begin{align}
q_{1}\cos(\gamma)+q_{2}\cos(\gamma+2\vartheta) & =R_{2}(\mathbf{a}),\nonumber \\
q_{1}\sin(\gamma)-q_{2}\sin(\gamma+2\vartheta) & =0.\label{eq:gamma_cond}
\end{align}
with $\vartheta=\alpha+\beta$. Here, the first equation gives the value for the second order parameter
while the second equation determines $\gamma$. A special solution
can be given if we set $q_{2}=0$, equivalently $q_{1}=1$ . Then
$\gamma=0$ and $\gamma=\pi$ would solve the equation above and thus
$2a_{i}=0$ for all $i=1,\dots,N$. Note that, for both values of
$\gamma$ the value for the complex second order parameter coincide
$Z_{2}(\mathbf{a})=1$. This solution corresponds to $R_{2}(\mathbf{a})=1$.
In any other case the last equation can be written in the form $\sin(\gamma-\nu)=0$
which has two solutions $\gamma=\nu,\nu+\pi.$ Both solution would
coincide while determining $2a_{i}$. Writing~(\ref{eq:gamma_cond})
as
\begin{align*}
\frac{1}{C}\left((q_{1}-(1-q_{1})\cos(2\vartheta))\sin(\gamma)-(1-q_{1})\sin(2\vartheta)\cos(\gamma)\right) & =0
\end{align*}
where the normalization constant $C$ is defined as 
\begin{align*}
C & =\sqrt{(q_{1}-(1-q_{1})\cos(2\vartheta))^{2}+(1-q_{1})^{2}\sin^{2}(2\vartheta)},
\end{align*}
yields the equations 
\begin{align}
\sin(\nu) & =\frac{\sin(2\vartheta)}{\sqrt{\left(\frac{q_{1}}{1-q_{1}}\right)^{2}+1-2\frac{q_{1}}{1-q_{1}}\cos(2\vartheta)}},\label{eq:SinNu}\\
\cos(\nu) & =\frac{q_{1}-(1-q_{1})\cos(2\vartheta)}{\sqrt{\left({1-q_{1}}\right)^{2}+q_1^2-2q_1({1-q_{1}})\cos(2\vartheta)}}.\nonumber 
\end{align}
Therefore, considering the inverse function $\arcsin:[-1,1]\to[-\pi/2,\pi/2]$
applied to~(\ref{eq:SinNu}) determines $\nu$ to be either $\nu'$
or $\pi-\nu'$, where $\nu':=\arcsin(\sin(\nu))$ and $\sin(\nu)$
as given in (\ref{eq:SinNu}). The second equation for $\cos(\nu)$
then fixes $\nu$ to take one of the values. Thus, $\gamma$ exists
and is unique for every $q_{1}\in[0,1)$. The Proposition is proved by Taking into account that
a finite number $N$ of oscillators $q_{1}$ takes values in the range
from $1/N$ to $(N-1)/N$ and defining $\gamma:=-\psi-\vartheta$.
\end{proof}

\subsection{Multi-cluster solutions of splay type}

\emph{Proof of Proposition (\ref{prop:MCSol_RotWave}).} We prove
by direct substitution. Plugging (\ref{eq:pRWMC}) and (\ref{eq:kRWMC})
into~(\ref{eq:KappaDGL_general}) the identity is obtained. Further,
substituting (\ref{eq:pRWMC}) and (\ref{eq:kRWMC}) into~(\ref{eq:KappaDGL_general})
we obtain
\begin{align*}
 \Omega_{\mu} &=\frac{1}{2N}\sum_{\nu=1}^{M}\rho_{\mu\nu}\sum_{j=1}^{N_{\mu}}\left(\cos(\alpha-\beta+\psi_{\mu\nu})\right.\\
 &\quad \left.-\cos(2(\Delta\Omega_{\mu\nu}t+a_{i,\mu}-a_{j,\nu})+\alpha+\beta-\psi_{\mu\nu})\right)\\
 &=\sum_{\nu=1}^{M}\rho_{\mu\nu}\left(\frac{n_{\nu}}{2}\cos(\alpha-\beta+\psi_{\mu\nu})-\frac{1}{2}\Re\left(e^{-\mathrm{i}(2\Delta\Omega_{\mu\nu}t+2a_{i,\mu}+\alpha+\beta-\psi_{\mu\nu})}Z_{2}(\mathbf{a}_{\nu})\right)\right).
\end{align*}
If $Z_{2}(\mathbf{a}_{\mu})=0$ for all $\mu=1,\dots,M$ then
\begin{align*}
\Omega_{\mu}=\frac{1}{N}\sum_{\nu=1}^{M}\sum_{j=1}^{N_{\mu}} & \rho_{\mu\nu}\cos(\alpha-\beta+\psi_{\mu\nu})
\end{align*}
which agrees with the system~(\ref{eq:MCRotWave_Omega}) for the
frequencies $\Omega_{\mu}$. On the contrary, assume that the multi-cluster
phase-locked solution (\ref{eq:pRWMC}) and (\ref{eq:kRWMC}) solve
the equation~(\ref{eq:PhiDGL_general}). Analog to Proposition~\ref{prop:1ClusterSolutions}
$\sum_{\nu=1}^{M}\rho_{\mu\nu}\Re\left(e^{-\mathrm{i}(2\Delta\Omega_{\mu\nu}t+2a_{i,\mu}+\alpha+\beta-\psi_{\mu\nu})}Z_{2}(\mathbf{a}_{\nu})\right)$
has to be independent of the oscillator index $i=1,\dots,N_{\mu}$
and $t\in\mathbb{R}$ for all $\mu=1,\dots,M$. Take any $\mu=1,\dots,M$
and suppose $Z_{2}(\mathbf{a}_{\nu})\ne0$ for $\nu\in A$ with $A\subseteq\{1,\dots,M\}$.
Then 
\begin{align}
\begin{split} & \sum_{\nu=1}^{M}\rho_{\mu\nu}\Re\left(e^{-\mathrm{i}(2\Delta\Omega_{\mu\nu}t+2a_{i,\mu}+\alpha+\beta-\psi_{\mu\nu})}Z_{2}(\mathbf{a}_{\nu})\right)=\\
 & \sum_{\nu\in A}\rho_{\mu\nu}\left(R_{2}(\mathbf{a}_{\nu})\cos(2\Delta\Omega_{\mu\nu}t+2a_{i,\mu}+\alpha+\beta-\psi_{\mu\nu}+\gamma_{\nu})\right)
\end{split}
\label{eq:OmegaCl_timedep}
\end{align}
where $\gamma_{\nu}$ is defined as in Proposition~\ref{prop:1ClusterSolutions},
see type 2 or 3. For fixed $\mu$ all frequency differences $\Delta\Omega_{\mu\nu}$
differ due to the assumption that the frequencies $\Omega_{\mu}$
are all pairwise different. This is why only terms with $\Delta\Omega_{\mu\nu}$
and $\Delta\Omega_{\mu\nu'}=-\Delta\Omega_{\mu\nu}$ ($\nu,\nu'\in A$)
are candidates to compensate each other in the right hand side of~(\ref{eq:OmegaCl_timedep}) to give a constant value for $\Omega_{\mu}$.
Therefore, the number of clusters with $Z_{2}(\mathbf{a}_{\nu})\ne0$
excluding the $\mu$-th cluster which is under consideration has to
be even, i.e., $|A\setminus\{\mu\}|$ even for all $\mu=1,\dots,M$.
This already yields that $|A|$ odd and $A=\{1,\dots,M\}$. Consider
now $\mu$ such that $\Omega_{\mu}=\min_{\nu\in{1,\dots,M}}\Omega_{\nu}$.
Then for every other $\nu\in\{1,\dots,M\}$ with $\Delta\Omega_{\mu\nu}<0$
there has to be $\nu'\in\{1,\dots,M\}$ so that $-\Delta\Omega_{\mu\nu}=\Delta\Omega_{\mu\nu'}=\Omega_{\mu}-\Omega_{\nu'}$.
Hence, $\Omega_{\nu'}<\Omega_{\mu}$ which contradicts that $\Omega_{\mu}=\min_{\nu\in{1,\dots,M}}\Omega_{\nu}$.
Therefore, for this choice of $\mu\in{1,\dots,M}$ the expression
in~(\ref{eq:OmegaCl_timedep}) cannot be constant contradicting
the assumption made in the beginning.

\subsection{Asymptotic expansions of multi-cluster solutions}

In this section we give an analytic description of multi-cluster solutions
in terms of an asymptotic expansion. We consider therefore the expansion
of $r$-th order together with a multi-time scale ansatz~\cite{VER06}
\begin{align}
\begin{split}\phi_{i,\mu}^{(r)} & (\epsilon,t):=\Omega_{\mu}^{(r)}(\tau_{0},\dots,\tau_{r})+a_{i,\mu}+\sum_{l=1}^{r}\epsilon^{l}p_{i,\mu;l}(t)\\
\kappa_{ij,\mu\nu}^{(r)} & (\epsilon,t):=\sum_{l=0}^{r}\epsilon^{l}k_{ij,\mu\nu;l}(t)
\end{split}
 &  & \begin{split}\mu,\nu & =1,\dots,M\\
i,j & =1,\dots,N_{\mu}
\end{split}
\label{eq:FormalExp_MultiTimeScale}
\end{align}
where $\Omega_{\mu}^{(r)}\in C^{1}(\mathbb{R}^{r+1})$ is a function
depending on the multi-time scales $\tau_{l}:=\epsilon^{l}t$. We show under under which conditions this expansion describes the time evolution for
the system (\ref{eq:PhiDGL_general})--(\ref{eq:KappaDGL_general}). 

The section is organized as follows. We first introduce some notations and state the main result. Then, we outline the strategy of the proof. We then prove some technical lemmata that will help us to prove the main result. For ease of notation, for the remainder of the section indices are used as follows. Small Latin letters $i$, $j$ in the subscript are oscillator indices while Greek letters $\mu$, $\nu$ represent cluster indices. These are separated by a comma. Two further indices, separated by semicolon, are the coefficient index and the mode index, respectively. The indices in superscript are either powers or the order for the expansion which are then written in parenthesis.

The following definition is introduced to handle the order of approximation.
\begin{definition}
Let $f:\mathbb{R}\times\mathbb{R}\to\mathbb{\mathbb{R}}$ and $g:\mathbb{R}\times\mathbb{R}\to\mathbb{\mathbb{R}}$
real functions. We define the following notations:
\end{definition}
\begin{enumerate}
\item $f(\epsilon,t)\in O(g(\epsilon,t))$ as $\epsilon\to0$ on the interval
$I\subseteq\mathbb{R}$ if for any $t\in I$ there exist $C(t)>0$
and $\epsilon_{0}(t)>0$ such that $|f(\epsilon,t)|<C(t)|g(\epsilon,t)|$
for all $\epsilon<\epsilon_{0}(t)$,
\item $f(\epsilon,t)\in o(g(\epsilon,t))$ as $\epsilon\to0$ on the interval
$I\subseteq\mathbb{R}$ if for any $t\in I$ and all $C>0$ there
exist $\epsilon_{0}(t)>0$ such that $|f(\epsilon,t)|<C|g(\epsilon,t)|$
for all $\epsilon<\epsilon_{0}(t)$.
\end{enumerate}
\begin{remark}
If the constants $C$ and $\epsilon_{0}$ can be chosen independently
of $t\in I$ we say that $f(\epsilon,t)\in O(g(\epsilon,t))$ (or
$f(\epsilon,t)\in o(g(\epsilon,t))$) as $\epsilon\to0$ uniformly
on $I$.
\end{remark}

In order to find an expressions for the asymptotic expansion of the coupling weights $\kappa$, we use the concept of the pullback attractor. It is defined as an nonempty, compact and invariant set and well known from the theory of nonautonomous dynamical systems. For our purposes, suppose we know the functions $\phi_i(t)$ for all $i=1,\dots,N$. Then, the differential equations~\eqref{eq:KappaDGL_general} is nonautonomous and can be solved explicitly by
\begin{align}
\kappa_{ij}(t) & :=\kappa_{ij,0}e^{-\epsilon(t-t_{0})}-\epsilon\int_{t_{0}}^{t}e^{-\epsilon(t-t')}\sin(\phi_{i}(t')-\phi_{j}(t')+\beta)\,\mathrm{d}t'
\end{align}
with $\kappa_{ij,0} \in[-1,1]$ for all $i,j=1,\dots,N$. For this, the pullback attractor $\mathcal{A}$ is given by the set
\begin{align}\label{eq:DefPullbackAttrSet}
	\mathcal{A}:=\bigcup_{t\in\mathbb{R}}\left\{(\kappa(t))_{*}\right\}
\end{align}
where
\begin{align}
(\kappa_{ij}(t))_{*} & :=-\lim_{t_{0}\to-\infty}\left(\epsilon\int_{t_{0}}^{t}e^{-\epsilon(t-t')}\sin(\phi_{i}(t')-\phi_{j}(t')+\beta)\,\mathrm{d}t'\right)\label{eq:DefPullBackAttrkappa}
\end{align}
for all $i,j\in\{1,\dots,N\}$. We remark the following important properties. For given functions $\phi_i$ the equations~\eqref{eq:KappaDGL_general} posses the compact absorbing set $G:=\{\kappa_{ij}: \kappa_{ij}\in[-1,1], i,j=1,\dots,N\}$. Hence, the pullback attractor exists, c.f. Theorem~3.18 in~\cite[pp. 45--46]{KLO11}, and is unique due to Proposition~3.8 \cite[p. 41]{KLO11}. Moreover, $(\kappa(t))_{*}$ is a solution for the nonautonomous system which can be shown by direct computation. We call $\kappa(t))_{*}$ the parametrization of the pullback attractor and use it in order to find an analytic expression for the (pseudo-)multi-cluster states. Note that we have already seen such parametrizations explicitly in~\eqref{eq:kPL} and~\eqref{eq:kRWMC}. For more details regarding nonautonomous systems and the pullback attractor we refer the reader to \cite{RAS06a,KLO11}.

We use the following notations for the sake of brevity.
\begin{align*}
\mathbf{M}&:=\left\{ \mathbf{m}=(m_{1},\dots,m_{M}):\,m_{1},\dots,m_{M}\in\mathbb{Z}\right\},\\
c_{\mathbf{m}} & :=c_{m_{1},\dots,m_{M}},\\
\Delta\Omega(\mathbf{m}) & :=\sum_{\mu=1}^{M}m_{\mu}\Omega_{\mu}.
\end{align*}
Furthermore, we say that two elements $\mathbf{m},\mathbf{n}\in\mathbf{M}$
are equivalent $\mathbf{m}\sim\mathbf{n}$ if and only if $\Delta\Omega(\mathbf{m})=\Delta\Omega(\mathbf{n})$.
The corresponding quotient space is denoted by $\tilde{\mathbf{M}}:=\mathbf{M}/_{\sim}$.
If $\Omega_{\mu}$ is considered as frequencies the equivalence relation
factors out all resonant linear combinations of those. Let us further
define $\mathbf{\tilde{M}}(f)$ as the set of all $(m_{1},\dots,m_{M})$
such that the function $f$ can be written as $f=\sum_{\mathbf{m}\in\mathbf{M}(f)}c_{\mathbf{m}}e^{\mathrm{i}\Delta\Omega(\mathbf{m})t}$
for some $c_{\mathbf{m}}\in\mathbb{C}$. Finally, we introduce the
shorthand notion $(m\mu n\nu):=(0,\dots,0,m_{\mu},0,\dots,0,m_{\nu},0,\dots,0)$
with $m_{\mu}=m$ and $m_{\nu}=n$ for further convenience if only
frequencies of two distinguished clusters are considered.

The main result on the asymptotic expansion for (pseudo-)multi-cluster solutions reads as follows.
\begin{proposition}
	\label{thm:FormalExpansionThm}Let $r\in\mathbb{\mathbb{N}}$. Suppose
	the system~(\ref{eq:PhiDGL_general})\textendash ~(\ref{eq:KappaDGL_general})
	possesses a pseudo multi-cluster solution $(\phi_{i,\mu},\kappa_{ij,\mu\nu})$
	with $\phi_{i,\mu}(\epsilon,t)=\Omega_{\mu}(\epsilon)t+a_{i,\mu}+s_{i,\mu}(\epsilon,t)$
	where $a_{i,\mu}\in\mathbb{T}^{1}$ and the coupling matrix $\kappa_{ij,\mu\nu}(\epsilon,t)$ is
	given as the parametrization of the pullback attractor defined in~(\ref{eq:DefPullBackAttrkappa}).
	Assume further that $M_{1}$ clusters are of antipodal type
	($2a_{i,\mu}={a}_{\mu}$) and $M_{2}$ are of splay type
	($R_{2}(\mathbf{a}_{\mu})=0$). Then, the $r$-th order asymptotic
	expansion of $\phi_{i,\mu}(\epsilon,t)$ for $t\in O(1/\epsilon^{r})$
	as $\epsilon\to0$ is given by
	\begin{align}
	\phi_{i,\mu}^{(r)}(\epsilon,t) & :=\Omega_{\mu,0}^{(r)}t+a_{i,\mu}+\sum_{l=1}^{r}\epsilon^{l}\left(\Omega_{\mu,l}^{(r)}t+p_{i,\mu;l}(t)\right) & \begin{split}\begin{split}\mu,\nu & =1,\dots,M\end{split}
	\end{split}
	\label{eq:GenFormalExpansion_phi}\\
	\kappa_{ij,\mu\nu}^{(r)}(\epsilon,t) & :=\sum_{l=0}^{r}\epsilon^{l}k_{ij,\mu\nu;l}(t), & \begin{split}i & =1,\dots,N_{\mu}\\
	j & =1,\dots,N_{\nu}
	\end{split}
	\label{eq:GenFormalExpansion_kappa}
	\end{align}
	where\\
	(i) all coefficients of the expansion can be found inductively;\\
	(ii) the first order approximation can be written as
	\begin{align*}
	\phi_{i,\mu}^{(1)} & =\left(\frac{n_{\mu}}{2}\left(\cos(\alpha-\beta)-\cos(\alpha+\beta)\right)-\epsilon\sum_{\overset{\nu=1}{\nu\ne\mu}}^{M}\frac{n_{\nu}}{2\Delta\Omega^{(1)}_{\mu\nu}}\sin(\alpha-\beta)\right)t+a_{i,\mu}+\epsilon p_{\mu;1}
	\end{align*}
	for $\mu=1,\dots,M_{1}$, and 
	\begin{align*}
	\phi_{i,\mu}^{(1)} & =\left(\frac{n_{\mu}}{2}\cos(\alpha-\beta)-\epsilon\sum_{\overset{\nu=1}{\nu\ne\mu}}^{M}\frac{n_{\nu}}{2\Delta\Omega^{(1)}_{\mu\nu}}\sin(\alpha-\beta)\right)t+a_{i,\mu}+\epsilon p_{i,\mu;1}(t)
	\end{align*}
	for $\mu=M_{1}+1,\dots,M$ with 
	\begin{align*}
	p_{\mu;1} & =-\sum_{\overset{\nu=1}{\nu\ne\mu}}^{M_{1}}\frac{n_{\nu}}{4\left(\Delta\Omega^{(1)}_{\mu\nu}\right)^{2}}\cos(2\Delta\Omega^{(1)}_{\mu\nu}t+{a}_{\mu}-{a}_{\nu}+\alpha+\beta) & \mu=1,\dots,M_{1}\\
	p_{i,\mu;1} & =-\sum_{\overset{\nu=1}{\nu\ne\mu}}^{M}\frac{n_{\nu}}{4\left(\Delta\Omega^{(1)}_{\mu\nu}\right)^{2}}\cos(2\Delta\Omega^{(1)}_{\mu\nu}t+2a_{i,\mu}-{a}_{\nu}+\alpha+\beta). & \mu=M_{1}+1,\dots,M
	\end{align*}
	The coupling weights are given by
	\begin{align*}
	\kappa_{ij,\mu\mu}^{(1)} & =-\sin(a_{i,\mu}-a_{j,\mu}+\beta), & \\
	\kappa_{ij,\mu\nu}^{(1)} & =\frac{\epsilon}{\Delta\Omega^{(1)}_{\mu\nu}}\cos(\Delta\Omega_{\mu\nu}t+a_{i,\mu}-a_{j,\nu}+\beta).
	\end{align*}
\end{proposition}

The proof makes use of several lemmas and is presented at the end of this section. Overall, we aim to describe the following particular form for the dynamical behaviour of the phase oscillators. The phases of the oscillators $\phi_{i,\mu}$ form a pseudo multi-cluster, c.f. definition~\ref{def:PhaseOscStates_PseudoMulCl}. Further, the bounded modulations for the phases of each oscillator are given as Taylor expansions in $\epsilon$ with periodic coefficients that can be expressed as Fourier sums with even modes. The strategy for the proof of the main result is as follows. 
\begin{enumerate}
	\item Assume that the phases of the oscillators are given as finite Taylor sums in $\epsilon$ with periodic coefficients, which are represented as finite Fourier sums. With this, the equations~\eqref{eq:KappaDGL_general} can be explicitly solved. Introducing the pull-back attractor provides us with a unique expression for the asymptotic solutions ($t\to\infty$). An explicit form for the expansion in $\epsilon$ of these solutions of the coupling weights $\kappa$ is given in Lemma~\ref{lem:KappaFormalExpansion}.
	\item The solutions of the coupling weights depend on the Fourier modes of the periodic expansion coefficients of the oscillators. An statement on their explicit dependence is provided by Lemma~\ref{lem:Lemma2}. More specifically, the expansion coefficients of the couplings weights only consist of even modes whenever the expansion coefficients for the phases of the oscillators consist only of odd modes.
	\item The expressions for the coupling weights in Lemma~\ref{lem:KappaFormalExpansion} are used to derive the explicit form for the phase oscillators. More specifically, the expansions coefficients are derived such that they satisfy the equations~\eqref{eq:PhiDGL_general}. Higher order terms which contribute to a linear growth are absorbed in an expansion for the oscillator frequencies.
	\item Finally, we find an iterative scheme to determine all expansion coefficients of the phases and coupling weights up to any order. Moreover, it is shown that the coefficients provided by the iterative scheme are consistent with the assumption on the expansion coefficients given in the beginning of the proof.
\end{enumerate}

To determining the asymptotic expansion explicitly, derivatives of composed function have to be carried out. The following Lemma provides us with a general form.
\begin{lemma}
\label{lem:FaaDiBrunoFormula}Suppose we have $n$-times differentiable
real functions $f$ and $g$. Let $T_{n}:=\left\{ (k_{1},\dots,k_{n}):\,1k_{1}+2k_{2}+\cdots+nk_{n}=n,k_{1},\dots,k_{n}\in\mathbb{N}_{0}\right\} $ denote the partitions of $n$. The composition $\left(f\circ g\right)$
is also $n$-times differentiable and the $n$th derivative can
be written as
\begin{align}
	D_{x}^{n}(f\circ g)(x_{0})=\sum_{(k_{1},\,\ldots\,,k_{n})\in T_{n}}\frac{n!}{k_{1}!\cdot\ \cdots\ \cdot k_{n}!}\bigl(D_{x}^{k_{1}+\ldots+k_{n}}f)\circ g(x_{0})\prod_{m=1}^{n}\biggl(\frac{D^{m}g}{m!}\biggr)^{k_{m}}\hspace*{-5pt}(x_{0}).\label{eq:FaaDiBrunoFormula}
\end{align}
\end{lemma}
\begin{proof}
See~\cite[pp. 95--96]{ABR88}.
\end{proof}
This expression for the $n$th-derivative is also known as the
Faà di Bruno formula.
\begin{lemma}
\label{lem:KappaFormalExpansion}Suppose the phase oscillators behave
as
\begin{align*}
\phi_{i,\mu}(\epsilon,t) & :=\Omega_{\mu}t+a_{i,\mu}+\sum_{l=1}^{r}\epsilon^{l}p_{i,\mu;l}(t), & \begin{split}\mu & =1,\dots,M\\
i & =1,\dots,N_{\mu}
\end{split}
\end{align*}
with $\Omega_{\mu}\in\mathbb{R}$, $a_{i,\mu}\in\mathbb{T}^{1}$ and
there are $m_{i,\mu;l}\in\mathbb{\mathbb{N}}_{0}$ such that the bounded
functions $p_{i,\mu;l}:\mathbb{R}\to\mathbb{R}$ are given by the
finite multi-Fourier series
\begin{align*}
	p_{i,\mu;l}(t)=\sum_{\mathbf{m}\in\tilde{\mathbf{M}}(p_{i,\mu;l})}c_{i,\mu;l;\mathbf{m}}e^{\mathrm{i}\Delta\Omega(\mathbf{m})t}
\end{align*}
with $|\tilde{\mathbf{M}}(p_{i,\mu;l})|=m_{i,\mu;l}$. 

Then, for $s\le r$ there exist functions $k_{ij,\mu\nu;l}(t)$ such that the asymptotic expansion of the pull-back attractor $\left(\kappa_{ij,\mu\nu}\right)_{*}$
defined in~(\ref{eq:DefPullBackAttrkappa}) can be written as
\begin{align}
	\left(\kappa_{ij,\mu\nu}\right)_{*} & =\sum_{l=0}^{s}\epsilon^{l}k_{ij,\mu\nu;l}(t)+\hat{R}_{ij,\mu\nu}(\epsilon,t) &  & \begin{split}\mu,\nu & =1,\dots,M\\
	i & =1,\dots,N_{\mu}\\
	j & =1,\dots,N_{\nu}
	\end{split}
	,\label{eq:PullbackKappaFormalExpansion}
\end{align}
where $\hat{R}_{ij,\mu\nu}(\epsilon,t)\in o(\epsilon^{s})$ uniformly
on $\mathbb{R}$ as $\epsilon\to0$ and $\kappa_{ij,\mu\nu}^{(s)}(t):=\sum_{l=0}^{s}\epsilon^{l}k_{ij,\mu\nu;l}(t)$. The $\kappa_{ij,\mu\nu}^{(s)}(t)$ is called the $s$-th-order asymptotic approximation of $\left(\kappa_{ij,\mu\nu}\right)_{*}$. Further, all  $k_{ij,\mu\nu;l}$ can also be written as a Fourier sum.
%
\end{lemma}
\begin{proof}
For fixed functions $\phi_{i,\mu}(t)$ the nonautonomous systems corresponding to the differential equation (\ref{eq:KappaDGL_general})
posses the following parametrization of the pullback attractor
\begin{align*}
	(\kappa_{ij,\mu\nu}(t))_{*} & :=-\lim_{t_{0}\to-\infty}\left(\epsilon\int_{t_{0}}^{t}e^{-\epsilon(t-t')}\sin(\phi_{i,\mu}(t')-\phi_{j,\nu}(t')+\beta)\,\mathrm{d}t'\right).
\end{align*}
Using~(\ref{eq:FaaDiBrunoFormula}) with $f=\sin(\phi_{i,\mu}-\phi_{j,\nu}+\beta)$
and $g=\Delta\Omega_{\mu\nu}t+a_{ij,\mu\nu}+\beta+\sum_{l=1}^{r}\epsilon^{l}p_{ij,\mu\nu;l}$
we can perform a Taylor expansion of $f(\epsilon,t)=\sin(\phi_{i,\mu}-\phi_{j,\nu}+\beta)$
around $\epsilon=0$. Due to Theorem 2.4.15 in~\cite[pp. 93--94]{ABR88} we get
\begin{multline}
\sin(\Delta\Omega_{\mu\nu}t+a_{ij,\mu\nu}+\beta+\sum_{l=1}^{r}\epsilon^{l}p_{ij,\mu\nu;l}(t))  =\sin(\Delta\Omega_{\mu\nu}t+a_{ij,\mu\nu}+\beta)\\
+\epsilon\cos(\Delta\Omega_{\mu\nu}t+a_{ij,\mu\nu}+\beta)p_{ij,\mu\nu;1}\nonumber \\
   +\dots+\epsilon^{s}R_{ij,\mu\nu}(\epsilon,t)\nonumber \\
  =:\sum_{l=0}^{s}\epsilon^{l}r_{ij,\mu\nu;l}(\beta,t)+\epsilon^{s}R_{ij,\mu\nu}(\epsilon,t)\label{eq:SineTaylorExpand}
\end{multline}
where the abbreviations $p_{ij,\mu\nu;l}:=p_{i,\mu;l}-p_{j,\nu;l}$
and $a_{ij,\mu\nu}:=a_{i,\mu}-a_{j,\nu}$ are used. Here, $R_{ij,\mu\nu}$
denotes the remainder of the Taylor expansion. The remainder
$R(\epsilon,t)\to0$ and $R_{ij,\mu\nu}(\epsilon,t)\in o(\epsilon)$ for all $t\in\mathbb{R}$ as $\epsilon\to0$. We get
\begin{multline}
	\left(\kappa_{ij,\mu\nu} (t)\right)_*=-\sum_{l=0}^{s}\epsilon^{l+1}\int_{-\infty}^{t}e^{-\epsilon(t-t')}r_{ij,\mu\nu;l}(\beta,t')\,\mathrm{d}t'\\+\epsilon^{s+1}\int_{t_{0}}^{t}e^{-\epsilon(t-t')}R_{ij,\mu\nu}(\epsilon,t')\,\mathrm{d}t'.\label{eq:kappaTaylorExpansion}
\end{multline}
In order to derive the expansion for $\left(\kappa_{ij,\mu\nu}\right)_{*}$
the integrals of the formula above have to be investigated. Faà di
Bruno's formula~(\ref{eq:FaaDiBrunoFormula}) provides us with
\begin{align}
r_{ij,\mu\nu;l}(\beta,t) & :=\sum_{(k_{1},\,\ldots\,,k_{l})\in T_{l}}\frac{\bigl(D^{k_{1}+\ldots+k_{l}}\sin)(\Delta\Omega_{\mu\nu}t+a_{ij,\mu\nu}+\beta)}{k_{1}!\cdot\ \cdots\ \cdot k_{l}!}\prod_{m=1}^{l}\biggl(p_{ij,\mu\nu;m}\biggr)^{k_{m}}.\label{eq:SinExpansionFaaDiBruno}
\end{align}
First, we conclude that the Taylor coefficients $r_{ij,\mu\nu;l}(\beta)$
can also be written in a (finite) multi-Fourier sum
\begin{align}
r_{ij,\mu\nu;l}(\beta,t)=\sum_{\mathbf{m}\in\tilde{\mathbf{M}}(r_{ij,\mu\nu;l}({\beta}))}d_{ij,\mu\nu;l;\mathbf{m}}({\beta})e^{\mathrm{i}\Delta\Omega(\mathbf{m})t}.\label{eq:SinExpansionFaaDiBrunoFourier}
\end{align}
Second, $d_{ij,\mu\nu;l;\mathbf{0}}\ne0$ if $\prod_{m=1}^{l}\biggl(p_{ij,\mu\nu,m}\biggr)^{k_{m}}$possess
a non vanishing term for $e^{\mathrm{i}\Delta\Omega_{\mu\nu}t}$.
With this, we are able to calculate the integrals in~(\ref{eq:kappaTaylorExpansion})
and get
\begin{align*}
\int_{-\infty}^{t}e^{-\epsilon(t-t')}r_{ij,\mu\nu;l}(\beta,t) & =\sum_{\mathbf{m}\in\tilde{\mathbf{M}}(r_{ij,\mu\nu;l}({\beta}))}d_{ij,\mu\nu;l;\mathbf{m}}({\beta})\int_{-\infty}^{t}e^{-\epsilon(t-t')}e^{\mathrm{i}\Delta\Omega(\mathbf{m})t'}\,\mathrm{d}t'\\
 & =\sum_{\mathbf{m}\in\tilde{\mathbf{M}}(r_{ij,\mu\nu;l}({\beta}))}\left(\frac{1}{\epsilon+\mathrm{i}\Delta\Omega(\mathbf{m})}\right)d_{ij,\mu\nu;l;\mathbf{m}}({\beta})e^{\mathrm{i}\Delta\Omega(\mathbf{m})t}.
\end{align*}
The last term in the equation~(\ref{eq:kappaTaylorExpansion}) is in $o(\epsilon^{s})$
which can be seen as follows. Since $R_{ij,\mu\nu}(\epsilon,t)\in o(\epsilon)$
as $\epsilon\to0$ for all $C>0$ there exist an $\epsilon_{0}(t)>0$
such that $\left|R_{ij,\mu\nu}(\epsilon,t)\right|<C\epsilon$ for
all $\epsilon<\epsilon_{0}(t)$. Due to the boundedness of all $p_{i,\mu;l}$
the remainder $R_{ij,\mu\nu}(\epsilon,t)$ is also bounded by some
positive number $\tilde{C}$. Thus, $R_{ij,\mu\nu}(\epsilon,t)\in o(\epsilon)$
uniformly on $\mathbb{R}$ as $\epsilon\to0$ , hence for all $C>0$
\begin{align*}
\left|\int_{t_{0}}^{t}e^{-\epsilon(t-t')}R(\epsilon,t')\,\mathrm{d}t'\right| & \le C\left|1-e^{\epsilon(t_{0}-t)}\right|.
\end{align*}
Finally, we end up with
\begin{align}
\left(\kappa_{ij,\mu\nu}\right)_{*} & (t)=-\sum_{l=0}^{s}\epsilon^{l}\sum_{\mathbf{m}\in\tilde{\mathbf{M}}(r_{ij,\mu\nu;l}^{\beta})}d_{ij,\mu\nu;l;\mathbf{m}}^{\beta}\frac{1}{1+\mathrm{i}\frac{\Delta\Omega(\mathbf{m})}{\epsilon}}e^{\mathrm{i}\Delta\Omega(\mathbf{m})t}+\tilde{R}_{ij,\mu\nu}(\epsilon,t)\label{eq:KappaPullbackExp}
\end{align}
where $\tilde{R}_{ij,\mu\nu}:=\left(\epsilon^{s+1}\int_{t_{0}}^{t}e^{-\epsilon(t-t')}R(\epsilon,t')\,\mathrm{d}t'\right)_{*}\in o(\epsilon^{s})$
uniformly on $\mathbb{R}$ as $\epsilon\to0$. By considering the
Laurent series
\begin{align*}
\frac{1}{1+\mathrm{i}\frac{\Delta\Omega(\mathbf{m})}{\epsilon}} & =-\sum_{n=1}^{\infty}i^{n}\left(\frac{\epsilon}{\Delta\Omega(\mathbf{m})}\right)^{n},
\end{align*}
which converges whenever $\epsilon<|\Delta\Omega(\mathbf{m})|,$ the
coefficients of the expansion $\left(\kappa_{ij,\mu\nu}\right)_{*}^{(s)} =\sum_{l=0}^{s}\epsilon^{l}k_{ij,\mu\nu;l}(t)$
are given by
\begin{align*}
\kappa_{ij,\mu\nu;0} & =-d_{ij,\mu\nu;0;\mathbf{0}}({\beta}),\\
\kappa_{ij,\mu\nu;l>0} & =-d_{ij,\mu\nu;l;\mathbf{0}}({\beta})+\sum_{n=0}^{l-1}\sum_{\mathbf{m}\in\tilde{\mathbf{M}}(r_{ij,\mu\nu;n}({\beta}))/\{0\}}i^{l-n}\frac{d_{ij,\mu\nu;n;\mathbf{m}}({\beta})}{\left(\Delta\Omega(\mathbf{m})\right)^{l-n}}e^{\mathrm{i}\Delta\Omega(\mathbf{m})t}.
\end{align*}
Note that, $\epsilon$ can always be chosen such that $\epsilon<|\Delta\Omega(\mathbf{m})|,$
since we consider the asymptotic limit $\epsilon\to0$. The coefficients
are determined via comparing the terms of both sides of the equation~(\ref{eq:KappaPullbackExp})
with respect to their order in $\epsilon$. In the
case \label{eq:KappaPullbackExpansion}$\mu\ne\nu$ we get $\kappa_{ij,\mu\nu;0}=0$.
All terms of order $O(\epsilon^{s+1})$ and $\tilde{R}_{ij,\mu\nu}(\epsilon,t)$
are summarized in $\hat{R}_{ij,\mu\nu}(\epsilon,t)\in o(\epsilon^{s})$
as $\epsilon\to0$.
%
%
\end{proof}
\begin{remark}
\label{rem:OmegaCondition_FormalExpansionKappa}Without considering
the asymptotic limit $\epsilon\to0$, the $s$-th order formal expansion
for $\kappa_{ij,\mu\nu}$ would have the form as it is given in Lemma~\ref{lem:KappaFormalExpansion}
under the condition that $\Delta\Omega(\mathbf{m})>\epsilon$ for
all $\mathbf{m}\in\bigcup_{l=1}^{s-1}\tilde{\mathbf{M}}(r_{ij,\mu\nu;l}^{\beta})$.
\end{remark}
\begin{lemma}\label{lem:Lemma2}
Suppose everything is given as in Lemma~\ref{lem:KappaFormalExpansion}. Then, if for all $\mu\in{1,\dots,M}$, $i\in\{1,\dots,N_{\mu}\}$ and
$l\in{1,\dots,r}$, $p_{i,\mu;l}(t)$ can be written completely in
terms of even modes, i.e., all $m_{1},\dots,m_{M}$ are even for $(m_{1},\dots,m_{M})\in\tilde{\boldsymbol{M}}(p_{i,\mu;l})$,
then for all $(n_{1},\dots,n_{M})\in\tilde{\boldsymbol{M}}(\kappa_{ij,\mu\nu}^{s}(t))$
holds: $n_{\lambda}$ are even for $\lambda\ne\mu,\nu$ and odd otherwise.
\end{lemma}
\begin{proof}
	It is fairly easy to verify that if $p_{i,\mu;l}(t)$ can be completely
	written in terms of even modes for all $i,\mu,l$, so can $p_{ij,\mu\nu;l}(t)$
	and moreover, the product $p_{ij,\mu\nu;l}\cdot p_{ij,\mu\nu;m}$.
	According to~(\ref{eq:SinExpansionFaaDiBruno}) $r_{ij,\mu\nu;l}$
	consists only of terms of the form $e^{\pm\mathrm{i}\Delta\Omega_{\mu\nu}t}\cdot e^{\mathrm{i}\Delta\Omega(m)t}$
	and hence $m_{\lambda}$ are even for every $\lambda\ne\mu,\nu$ and
	odd otherwise for all $(m_{1},\dots,m_{M})\in\tilde{\mathbf{M}}(r_{ij,\mu\nu;l})$.
	Since integration by time (\ref{eq:kappaTaylorExpansion}) does not
	make any changes in the modes, the same holds for $k_{ij,\mu\nu;l}(t)$
	and hence $\kappa_{ij,\mu\nu}^{s}(t)$.
\end{proof}

Now, we have everything which is needed to proof the main result Prop~\ref{thm:FormalExpansionThm}.
\begin{proof}
Note, whenever we write $\Delta\Omega(\mathbf{m})$, $\Delta\Omega^{(r)}(\mathbf{m})$ is meant. We omit the superscript for the sake of readability.
(i) Combing the formal time derivative of the first equation in~(\ref{eq:FormalExp_MultiTimeScale}),
the system equations~(\ref{eq:PhiDGL_general})\textendash (\ref{eq:KappaDGL_general})
and Lemma~\ref{lem:KappaFormalExpansion} we get
\begin{align}
\dot{\phi}_{i,\mu}^{(r)} & =\sum_{l=0}^{r}\epsilon^{l}\frac{\partial\Omega_{\mu}^{(r)}}{\partial\tau_{l}}+\sum_{l=1}^{r}\epsilon^{l}\dot{p}_{i,\mu;l}(t)=-\frac{1}{N}\sum_{\nu=1}^{M}\sum_{j=1}^{N_{\nu}}\sum_{l=0}^{r}\sum_{n=0}^{r}\epsilon^{l+n}k_{ij,\mu\nu;l}(t)r_{ij,\mu\nu,n}({\alpha},t).\label{eq:FormalExpan_phiDer}
\end{align}
Assume that $p_{i,\mu;l}(t)=\sum_{\mathbf{m}\in\tilde{\mathbf{M}}(p_{i,\mu;l})}c_{i,\mu;l;\mathbf{m}}e^{\mathrm{i}\Delta\Omega(\mathbf{m})t}$
with $|\tilde{\mathbf{M}}(p_{i,\mu;l})|\in\mathbb{N}$. We get
\begin{align}
\frac{\partial\Omega_{\mu}^{(r)}}{\partial t} & =\frac{1}{N}\sum_{\nu=1}^{M}\sum_{j=1}^{N_{\nu}}\sum_{\mathbf{m}\in\tilde{\mathbf{M}}(r_{ij,\mu\nu;0}({\alpha}))}d_{ij,\mu\nu;0;\mathbf{m}}({\alpha})d_{ij,\mu\nu;0;\mathbf{0}}({\beta})e^{\mathrm{i}\Delta\Omega(\mathbf{m})t},\nonumber \\
\frac{\partial\Omega_{\mu}^{(r)}}{\partial\tau_{l}}+\dot{p}_{i,\mu;l}(t) & =-\frac{1}{N}\sum_{\nu=1}^{M}\sum_{j=1}^{N_{\nu}}\sum_{m=0}^{l}k_{ij,\mu\nu;m}r_{ij,\mu\nu,l-m}({\alpha})\label{eq:GenPerturbationOmegaAndP}
\end{align}
by comparing both sides of the equation~(\ref{eq:FormalExpan_phiDer}) with
respect to the order of $\epsilon$. Due to Lemma~\ref{lem:KappaFormalExpansion}
\begin{align*}
\frac{\partial\Omega_{\mu}^{(r)}(\tau_{0},\dots,\tau_{m})}{\partial t} & =\frac{1}{N}\sum_{j=1}^{N_{\mu}}d_{ij,\mu\mu;0;\mathbf{0}}({\alpha})d_{ij,\mu\mu;0;\mathbf{0}}({\beta})=:\Omega_{\mu,0}\in\mathbb{R}.
\end{align*}
This equation can be solved by $\tilde{\Omega}_{\mu}=\Omega_{\mu,0}t+\tilde{\Omega}_{\mu,0}(\tau_{1},\dots,\tau_{m})$.
Due to our assumptions the right hand side of equation~(\ref{eq:GenPerturbationOmegaAndP})
can be written as
\begin{multline*}
-\sum_{\nu=1}^{M}\sum_{j=1}^{N_{\nu}}\sum_{m=0}^{l}k_{ij,\mu\nu;m}r_{ij,\mu\nu,l-m}({\alpha})=\\
\sum_{\nu=1}^{M}\sum_{j=1}^{N_{\nu}}\sum_{m=0}^{l}\left[d_{ij,\mu\nu;m;\mathbf{0}}({\beta})-\sum_{n=0}^{m-1}\sum_{\mathbf{m}\in\tilde{\mathbf{M}}(r_{ij,\mu\nu;n}({\beta}))/\{0\}}i^{m-n}\frac{d_{ij,\mu\nu;n;\mathbf{m}}({\beta})}{\left(\Delta\Omega(\mathbf{m})\right)^{m-n}}e^{\mathrm{i}\Delta\Omega(\mathbf{m})t}\right]\times\\
\left[d_{ij,\mu\nu;l-m;\mathbf{0}}({\alpha})+\sum_{\mathbf{m}\in\tilde{\mathbf{M}}(r_{ij,\mu\nu;l-m}({\alpha}))/\{0\}}d_{ij,\mu\nu;l-m;\mathbf{m}}({\alpha})e^{\mathrm{i}\Delta\Omega(\mathbf{m})t}\right].
\end{multline*}
By using 1.) and 2.) of Lemma~\ref{lem:KappaFormalExpansion} we
find
\begin{multline*}
\frac{\partial\Omega_{\mu}^{(r)}}{\partial\tau_{l}} =\frac{1}{N}\sum_{\nu=1}^{M}\sum_{j=1}^{N_{\nu}}\sum_{m=0}^{l}\left[d_{ij,\mu\nu;m;\mathbf{0}}({\beta})d_{ij,\mu\nu;l-m;\mathbf{0}}({\alpha})\right.\\
\left.-\sum_{n=0}^{m-1}\sum_{\mathbf{m}\in\tilde{\mathbf{M}}(r_{ij,\mu\nu;n}({\beta}))/\{0\}}i^{m-n}\frac{d_{ij,\mu\nu;n;\mathbf{m}}({\beta})d_{ij,\mu\nu;l-m;\mathbf{-m}}({\alpha})}{\left(\Delta\Omega(\mathbf{m})\right)^{m-n}}\right],
\end{multline*}
\begin{multline*}
	\dot{p}_{i,\mu;l}  =\\
	\frac{1}{N}\sum_{\nu=1}^{M}\sum_{j=1}^{N_{\nu}}\sum_{m=0}^{l}\left[d_{ij,\mu\nu;m;\mathbf{0}}({\beta})\left(\sum_{\mathbf{m}\in\tilde{\mathbf{M}}(r_{ij,\mu\nu;l-m}({\beta}))/\{0\}}d_{ij,\mu\nu;l-m;\mathbf{m}}({\alpha})e^{\mathrm{i}\Delta\Omega(\mathbf{m})t}\right)\right.\\
	 \left.-d_{ij,\mu\nu;l-m;\mathbf{0}}({\alpha})\left(\sum_{n=0}^{m-1}\sum_{\mathbf{m}\in\tilde{\mathbf{M}}(r_{ij,\mu\nu;n}({\beta}))/\{0\}}i^{m-n}\frac{d_{ij,\mu\nu;n;\mathbf{m}}({\beta})}{\left(\Delta\Omega(\mathbf{m})\right)^{m-n}}e^{\mathrm{i}\Delta\Omega(\mathbf{m})t}\right)\right.\\
	 \left.-\sum_{n=0}^{m-1}\sum_{\overset{\scriptstyle \mathbf{m}\in\tilde{\mathbf{M}}(r_{ij,\mu\nu;n}({\beta}))/\{0\}}{\mathbf{n}\in\tilde{\mathbf{M}}(r_{ij,\mu\nu;l-m}({\alpha}))/\{0,-\mathbf{m}\}}}i^{m-n}\frac{d_{ij,\mu\nu;n;\mathbf{m}}({\beta})d_{ij,\mu\nu;l-m;\mathbf{n}}({\alpha})}{\left(\Delta\Omega(\mathbf{m})\right)^{m-n}}e^{\mathrm{i}\Delta\Omega(\mathbf{m}+\mathbf{n})t}\right].
\end{multline*}
Note that we use the multi-time scale function $\Omega_{\mu}^{(r)}$
to deal with all terms of the expansion describing a linear growth.
All the other terms are considered to determine the behaviour of $p_{i,\mu;l}.$
With this ansatz we are able to maintain the boundedness of $p_{i,\mu;l}$
while letting $\Omega_{\mu}^{(r)}$ alone describing unbounded behaviour
in $t\in\mathbb{R}$. Note further that $\Omega_{\mu}^{(r)}$ can be directly
computed if all functions $p_{i,\mu;k\le l}(t)$ are known. Thus,
we finally end up with
\begin{align*}
\Omega_{\mu}^{(r)} & =\sum_{l=0}^{r}\epsilon^{l}\Omega_{\mu,l}t.
\end{align*}
We assume now that for all $i,\mu$ and $l>1$, $p_{i,\mu;l}$ can
be written completely in terms of even modes, c.f., Lemma~\ref{lem:Lemma2}.
In particular, $(\mu\nu)\notin M(p_{i,\mu;l})$. Thus, $d_{ij,\mu\nu;l;\mathbf{0}}({\beta})=0$
by~(\ref{eq:SinExpansionFaaDiBruno}) for all $\mu,\nu=1,\dots,M$,
$i=1,\dots,N_{\mu}$, $j=1,\dots,N_{\nu}$ and $l\ge1$ and 
\begin{multline*}
\dot{p}_{i,\mu;l} =-\frac{1}{N}\sum_{\nu=1}^{M}\sum_{j=1}^{N_{\nu}}\sum_{\overset{\scriptstyle m=1}{n=0}}^{l,m-1}\sum_{\overset{\scriptstyle \mathbf{m}\in\tilde{\mathbf{M}}(r_{ij,\mu\nu;n}({\beta}))/\{0\}}{\mathbf{n}\in\tilde{\mathbf{M}}(r_{ij,\mu\nu;l-m}({\alpha}))/\{0,-\mathbf{m}\}}}\hspace{-10pt}i^{m-n}\times\\
\frac{d_{ij,\mu\nu;n;\mathbf{m}}({\beta})d_{ij,\mu\nu;l-m;\mathbf{n}}({\alpha})}{\left(\Delta\Omega(\mathbf{m})\right)^{m-n}}e^{\mathrm{i}\Delta\Omega(\mathbf{m}+\mathbf{n})t}.
\end{multline*}
Hence, we get an equation to determine the value of $p_{i,\mu;l}$
inductively. Due to Lemma~\ref{lem:Lemma2}, we
know that if all $p_{i,\mu;l}$ can be written in terms of even modes
then $m_{\mu}$ and $m_{\nu}$ are odd for all $(m_{1},\dots,m_{M})\in\tilde{\mathbf{M}}(r_{ij,\mu\nu;l}({\alpha}))$.
Therefore $p_{i,\mu;l}$ can be written in terms of even modes. This
is consistent with our assumption that $p_{i,\mu;l}$ can be written
in terms of even modes. Consider further 
\begin{align*}
\dot{p}_{i,\mu;1} & =-\frac{1}{N}\sum_{\nu=1}^{M}\sum_{j=1}^{N_{\nu}}\sum_{\overset{\scriptstyle \mathbf{m}\in\tilde{\mathbf{M}}(r_{ij,\mu\nu;0}({\beta}))/\{0\}}{\mathbf{n}\in\tilde{\mathbf{M}}(r_{ij,\mu\nu;0}({\alpha}))/\{0,-\mathbf{m}\}}}i\frac{d_{ij,\mu\nu;0;\mathbf{m}}({\beta})d_{ij,\mu\nu;0;\mathbf{n}}({\alpha})}{\left(\Delta\Omega(\mathbf{m})\right)^{m-n}}e^{\mathrm{i}\Delta\Omega(\mathbf{m}+\mathbf{n})t}.
\end{align*}
The expression for $p_{i,\mu,1}$ can be found by integration
\begin{align*}
p_{i,\mu;1}(t) & =-\frac{1}{N}\sum_{\nu=1}^{M}\sum_{j=1}^{N_{\nu}}\sum_{\overset{\scriptstyle \mathbf{m}\in\tilde{\mathbf{M}}(r_{ij,\mu\nu;0}({\beta}))/\{0\}}{\mathbf{n}\in\tilde{\mathbf{M}}(r_{ij,\mu\nu;0}({\alpha}))/\{0,-\mathbf{m}\}}}\frac{d_{ij,\mu\nu;0;\mathbf{m}}({\beta})d_{ij,\mu\nu;0;\mathbf{n}}({\alpha})}{\left(\Delta\Omega(\mathbf{m})\right)\Delta\Omega(\mathbf{m}+\mathbf{n})}e^{\mathrm{i}\Delta\Omega(\mathbf{m}+\mathbf{n})t}.
\end{align*}
 Since $p_{i,\mu;1}$ can be written as a Fourier sum the same holds
true for all $p_{i,\mu,l}$ by induction. This is consistent with
our assumption for $p_{i,\mu,l}$ in the beginning of this proof.
The expressions $k_{ij,\mu\nu;l}$ follow from Lemma~\ref{lem:KappaFormalExpansion}.
Furthermore, analog to Lemma~\ref{lem:KappaFormalExpansion},
from Theorem 2.4.15 in~\cite[pp. 93--94]{ABR88} we conclude $\kappa_{ij,\mu\nu}(\epsilon,t)-\kappa_{ij,\mu\nu}^{r}(\epsilon,t)\in O(\epsilon^{r})$
and $\phi_{i,\mu}(\epsilon,t)-\phi_{i,\mu}^{(r)}(\epsilon,t)\in O(\epsilon^{r})$
for $t\in O(1/\epsilon^{r})$ as $\epsilon\to0$.\\
(ii) To achieve this result we apply now (i) which allows for iteratively determining
the function appearing in the asymptotic expansion. 

$0$-th order: For the expansion of the sine-function from equation~(\ref{eq:SinExpansionFaaDiBruno})
we find $r_{ij,\mu\mu;0}({\beta})=\sin(a_{ij,\mu\mu}+\beta)=d_{ij,\mu\mu;0;\mathbf{0}}({\beta})$
and
\begin{align*}
r_{ij,\mu\nu;0}({\beta}) & =\sin(\Delta\Omega^{(1)}_{\mu\nu}t+a_{ij,\mu\nu}+\beta)=d_{ij,\mu\nu;0;(\mu\nu)}({\beta})e^{\mathrm{i}\Delta\Omega^{(1)}_{\mu\nu}t}+c.c.
\end{align*}
where $d_{ij,\mu\nu;0;(\mu\nu)}({\beta}):=(1/2\mathrm{i})e^{\mathrm{i}\left(a_{ij,\mu\nu}+\beta\right)}$
and $c.c.$ stands for complex conjugated. Hence, for the coupling
matrix we find
\begin{align*}
\kappa_{ij,\mu\mu;0} & =-\sin(a_{ij,\mu\mu}+\beta).
\end{align*}
Depending on the cluster the zero-th order approximation for the frequencies
read
\begin{align*}
\Omega_{\mu,0} & =\frac{1}{N}\sum_{j=1}^{N_{\mu}}\sin(a_{ij,\mu\mu}+\beta)\sin(a_{ij,\mu\mu}+\beta)=\frac{n_{\mu}}{2}\left(\cos(\alpha-\beta)-\cos(\alpha+\beta)\right)
\end{align*}
for all $\mu=1,\dots,M_{1}$ and analogously $\Omega_{\mu,0}=\frac{n_{\mu}}{2}\cos(\alpha-\beta)$
for all $\mu=M_{1}+1,\dots,M$.

$1$-th order: Since we know the $0$-th order expansion we are able
to calculate the next order. We get
\begin{align*}
\Omega_{\mu,1} & =-\frac{1}{N}\sum_{\overset{\nu=1}{\nu\ne\mu}}^{M_{1}}\sum_{j=1}^{N_{\nu}}\sum_{\mathbf{m}\in\{(\mu\nu),-(\mu\nu)\}}i\frac{d_{ij,\mu\nu;0;\mathbf{m}}({\beta})d_{ij,\mu\nu;0;\mathbf{-m}}({\alpha})}{\Delta\Omega^{(1)}(\mathbf{m})}\\
 & =\frac{1}{2N}\sum_{\overset{\nu=1}{\nu\ne\mu}}^{M}\sum_{j=1}^{N_{\nu}}\frac{1}{2\mathrm{i}}\left(\frac{e^{\mathrm{i}\left(a_{ij,\mu\nu}+\beta\right)}e^{-\mathrm{i}\left(a_{ij,\mu\nu}+\alpha\right)}}{\Delta\Omega^{(1)}_{\mu\nu}}-\frac{e^{-\mathrm{i}\left(a_{ij,\mu\nu}+\beta\right)}e^{\mathrm{i}\left(a_{ij,\mu\nu}+\alpha\right)}}{\Delta\Omega^{(1)}_{\mu\nu}}\right)\\
 & =\frac{1}{2N}\sum_{\overset{\nu=1}{\nu\ne\mu}}^{M}\sum_{j=1}^{N_{\nu}}\frac{1}{2\mathrm{i}}\left(\frac{e^{\mathrm{i}(\beta-\alpha)}}{\Delta\Omega^{(1)}_{\mu\nu}}-\frac{e^{-\mathrm{i}(\beta-\alpha)}}{\Delta\Omega^{(1)}_{\mu\nu}}\right)=-\sum_{\overset{\nu=1}{\nu\ne\mu}}^{M}\frac{n_{\nu}}{2\Delta\Omega^{(1)}_{\mu\nu}}\sin(\alpha-\beta).
\end{align*}
For all $\mu=1,\dots,M_{1}$ we get
\begin{align*}
\dot{p}_{\mu;1} & =-\frac{1}{N}\sum_{\overset{\nu=1}{\nu\ne\mu}}^{M_{1}}\sum_{j=1}^{N_{\nu}}\sum_{\mathbf{m}\in\{(\mu\nu),-(\mu\nu)\}}\left(i\frac{d_{ij,\mu\nu;0;\mathbf{m}}({\beta})d_{ij,\mu\nu;0;\mathbf{m}}({\alpha})}{\Delta\Omega^{(1)}(\mathbf{m})}\right)e^{\mathrm{i}2\Delta\Omega^{(1)}(\mathbf{m})t}\\
 & +\sum_{\nu=M_{1}+1}^{M}\sum_{j=1}^{N_{\nu}}\sum_{\mathbf{m}\in\{(\mu\nu),-(\mu\nu)\}}\left(i\frac{d_{ij,\mu\nu;0;\mathbf{m}}({\beta})d_{ij,\mu\nu;0;\mathbf{m}}({\alpha})}{\Delta\Omega^{(1)}(\mathbf{m})}\right)e^{\mathrm{i}2\Delta\Omega^{(1)}(\mathbf{m})t}\\
 & =-\frac{1}{2}\sum_{\overset{\nu=1}{\nu\ne\mu}}^{M}\frac{n_{\nu}}{2\mathrm{i}}\left(\frac{e^{\mathrm{i}({a}_{\mu}-{a}_{\nu}+\alpha+\beta)}}{\Delta\Omega^{(1)}_{\mu\nu}}\right)e^{\mathrm{i}2\Delta\Omega^{(1)}_{\mu\nu}t}-\frac{n_{\nu}}{2\mathrm{i}}\left(\frac{e^{-\mathrm{i}({a}_{\mu}-{a}_{\nu}+\alpha+\beta)}}{\Delta\Omega^{(1)}_{\mu\nu}}\right)e^{-\mathrm{i}2\Delta\Omega^{(1)}_{\mu\nu}t}
\end{align*}
and for all $\mu=M_{1}+1,\dots,M$ we get
\begin{align*}
\dot{p}_{i,\mu;1} & =-\frac{1}{N}\sum_{\nu=1}^{M_{1}}\sum_{j=1}^{N_{\nu}}\sum_{\mathbf{m}\in\{(\mu\nu),-(\mu\nu)\}}\left(i\frac{d_{ij,\mu\nu;0;\mathbf{m}}({\beta})d_{ij,\mu\nu;0;\mathbf{m}}({\alpha})}{\Delta\Omega^{(1)}(\mathbf{m})}\right)e^{\mathrm{i}2\Delta\Omega^{(1)}(\mathbf{m})t}\\
 & +\sum_{\overset{\nu=M_{1}+1}{\nu\ne\mu}}^{M}\sum_{j=1}^{N_{\nu}}\sum_{\mathbf{m}\in\{(\mu\nu),-(\mu\nu)\}}\left(i\frac{d_{ij,\mu\nu;0;\mathbf{m}}({\beta})d_{ij,\mu\nu;0;\mathbf{m}}({\alpha})}{\Delta\Omega^{(1)}(\mathbf{m})}\right)e^{\mathrm{i}2\Delta\Omega^{(1)}(\mathbf{m})t}\\
 & =-\frac{1}{2}\sum_{\nu=1}^{M}\frac{n_{\nu}}{2\mathrm{i}}\left(\frac{e^{\mathrm{i}(a_{i,\mu}-{a}_{\nu}+\alpha+\beta)}}{\Delta\Omega^{(1)}_{\mu\nu}}\right)e^{\mathrm{i}2\Delta\Omega^{(1)}_{\mu\nu}t}-\frac{n_{\nu}}{2\mathrm{i}}\left(\frac{e^{-\mathrm{i}(a_{i,\mu}-{a}_{\nu}+\alpha+\beta)}}{\Delta\Omega^{(1)}_{\mu\nu}}\right)e^{-\mathrm{i}2\Delta\Omega^{(1)}_{\mu\nu}t}.
\end{align*}
Thus, solving this fairly easy differential equation
the following expression is obtained for $\mu=1,\dots,M_{1}$
\begin{align*}
p_{\mu;1} & =\frac{1}{4}\sum_{\overset{\nu=1}{\nu\ne\mu}}^{M}\frac{n_{\nu}}{2}\left(\frac{e^{\mathrm{i}({a}_{\mu}-{a}_{\nu}+\alpha+\beta)}}{\left(\Delta\Omega^{(1)}_{\mu\nu}\right)^{2}}\right)e^{\mathrm{i}2\Delta\Omega^{(1)}_{\mu\nu}t}+\frac{n_{\nu}}{2}\left(\frac{e^{-\mathrm{i}({a}_{\mu}-{a}_{\nu}+\alpha+\beta)}}{\left(\Delta\Omega^{(1)}_{\mu\nu}\right)^{2}}\right)e^{-\mathrm{i}2\Delta\Omega^{(1)}_{\mu\nu}t}.
\end{align*}
 Analogously we find the expression for $p_{i,\mu;1}$ with $\mu=M_{1}+1,\dots,M$.
For the coupling matrix we get
\begin{align*}
\kappa_{ij,\mu\nu;1} & =\sum_{\mathbf{m}\in\{(\mu\nu),-(\mu\nu)\}}\left(i\frac{d_{ij,\mu\nu;0;\mathbf{m}}({\beta})}{\Delta\Omega^{(1)}(\mathbf{m})}\right)e^{\mathrm{i}\Delta\Omega^{(1)}(\mathbf{m})t}\\
&=\frac{e^{\mathrm{i}\left(a_{ij,\mu\nu}+\beta\right)}}{2\Delta\Omega^{(1)}_{\mu\nu}}e^{\mathrm{i}\Delta\Omega^{(1)}_{\mu\nu}t}+\frac{e^{-\mathrm{i}\left(a_{ij,\mu\nu}+\beta\right)}}{2\Delta\Omega^{(1)}_{\mu\nu}}e^{-\mathrm{i}\Delta\Omega^{(1)}_{\mu\nu}t}.
\end{align*}
\end{proof}
\section*{Acknowledgments}
We thank Vladimir Nekorkin for stimulating discussions.
\bibliographystyle{siamplain}
\bibliography{refs_MC_adaptive_networks}

\begin{thebibliography}{10}

\bibitem{ABB00}
{\sc L.~F. Abbott and S.~Nelson}, {\em Synaptic plasticity: taming the beast},
  Nat. Neurosci., 3 (2000), pp.~1178--1183.

\bibitem{ABR88}
{\sc R.~Abraham, J.~E. Marsden, and T.~Ratiu}, {\em Manifolds, Tensor Analysis,
  and Applications}, Springer, New York, 1988.

\bibitem{ABR04}
{\sc D.~M. Abrams and S.~H. Strogatz}, {\em Chimera states for coupled
  oscillators}, Phys. Rev. Lett., 93 (2004), 174102 (4~pages),
  \url{https://doi.org/10.1103/physrevlett.93.174102}.

\bibitem{ACE05a}
{\sc J.~A. Acebr{\'o}n, L.~L. Bonilla, C.~J.~P. Vicente, F.~Ritort, and
  R.~Spigler}, {\em The kuramoto model: A simple paradigm for synchronization
  phenomena}, Rev. Mod. Phys., 77 (2005), pp.~137--185.

\bibitem{AOK15}
{\sc T.~Aoki}, {\em Self-organization of a recurrent network under ongoing
  synaptic plasticity}, Neural Networks, 62 (2015), pp.~11--19,
  \url{https://doi.org/10.1016/j.neunet.2014.05.024}.

\bibitem{AOK09}
{\sc T.~Aoki and T.~Aoyagi}, {\em Co-evolution of phases and connection
  strengths in a network of phase oscillators}, Phys. Rev. Lett., 102 (2009),
  034101 (4~pages), \url{https://doi.org/10.1103/physrevlett.102.034101}.

\bibitem{AOK11}
{\sc T.~Aoki and T.~Aoyagi}, {\em Self-organized network of phase oscillators
  coupled by activity-dependent interactions}, Phys. Rev. E, 84 (2011), 066109
  (14~pages), \url{https://doi.org/10.1103/physreve.84.066109}.

\bibitem{ASH16a}
{\sc P.~Ashwin, C.~Bick, and O.~Burylko}, {\em Identical phase oscillator
  networks: Bifurcations, symmetry and reversibility for generalized coupling},
  Front. Appl. Math. Stat., 2 (2016),
  \url{https://doi.org/10.3389/fams.2016.00007}.

\bibitem{ASH08}
{\sc P.~Ashwin, O.~Burylko, and Y.~Maistrenko}, {\em Bifurcation to
  heteroclinic cycles and sensitivity in three and four coupled phase
  oscillators}, Physica D, 237 (2008),
  \url{https://doi.org/doi:10.1016/j.physd.2007.09.015}.

\bibitem{AVA18}
{\sc V.~Avalos-Gayt\'an, J.~A. Almendral, I.~Leyva, F.~Battiston, V.~Nicosia,
  V.~Latora, and S.~Boccaletti}, {\em Emergent explosive synchronization in
  adaptive complex networks}, Phys. Rev. E, 97 (2018), 042301 (7~pages),
  \url{https://doi.org/10.1103/physreve.97.042301}.

\bibitem{BLA13}
{\sc K.~Blaha, J.~Lehnert, A.~Keane, T.~Dahms, P.~H{\"o}vel, E.~Sch{\"o}ll, and
  J.~L. Hudson}, {\em Clustering in delay-coupled smooth and relaxational
  chemical oscillators}, Phys. Rev. E, 88 (2013), 062915 (9~pages),
  \url{https://doi.org/10.1103/physreve.88.062915}.

\bibitem{BLA99a}
{\sc B.~Blasius, A.~Huppert, and L.~Stone}, {\em Complex dynamics and phase
  synchronization in spatially extended ecological systems}, Nature, 399
  (1999), pp.~354--359.

\bibitem{BOY04}
{\sc S.~Boyd and L.~Vandenberghe}, {\em Convex Optimization}, Cambridge
  University Press, 2004.

\bibitem{BUR11}
{\sc O.~Burylko and A.~Pikovsky}, {\em Desynchronization transitions in
  nonlinearly coupled phase oscillators}, Physica D, 240 (2011),
  \url{https://doi.org/10.1016/j.physd.2011.05.016}.

\bibitem{CAP08a}
{\sc N.~Caporale and Y.~Dan}, {\em Spike timing$-$dependent plasticity: A
  hebbian learning rule}, Annu. Rev. Neurosci., 31 (2008), pp.~25--46,
  \url{https://doi.org/10.1146/annurev.neuro.31.060407.125639}.

\bibitem{CHO09}
{\sc C.~U. Choe, T.~Dahms, P.~H{\"o}vel, and E.~Sch{\"o}ll}, {\em Controlling
  synchrony by delay coupling in networks: from in-phase to splay and cluster
  states}, Phys. Rev. E, 81 (2010), 025205(R) (4~pages),
  \url{https://doi.org/10.1103/physreve.81.025205}.

\bibitem{DAH12}
{\sc T.~Dahms, J.~Lehnert, and E.~Sch{\"o}ll}, {\em Cluster and group
  synchronization in delay-coupled networks}, Phys. Rev. E, 86 (2012), 016202
  (10~pages), \url{https://doi.org/10.1103/physreve.86.016202}.

\bibitem{GRO08a}
{\sc T.~Gross and B.~Blasius}, {\em Adaptive coevolutionary networks: a
  review}, J. R. Soc. Interface, 5 (2008), pp.~259--271,
  \url{https://doi.org/10.1098/rsif.2007.1229},
  \url{https://arxiv.org/abs/http://rsif.royalsocietypublishing.org/content/5/20/259.full.pdf+html}.

\bibitem{GUS15a}
{\sc A.~Gushchin, E.~Mallada, and A.~Tang}, {\em Synchronization of
  phase-coupled oscillators with plastic coupling strength}, in Information
  Theory and Applications Workshop, 2015, pp.~291--300,
  \url{https://doi.org/10.1109/ita.2015.7309003}.

\bibitem{HAG12}
{\sc A.~M. Hagerstrom, T.~E. Murphy, R.~Roy, P.~H{\"o}vel, I.~Omelchenko, and
  E.~Sch{\"o}ll}, {\em Experimental observation of chimeras in coupled-map
  lattices}, Nature Phys., 8 (2012), pp.~658--661,
  \url{https://doi.org/10.1038/nphys2372}.

\bibitem{HAM07}
{\sc C.~Hammond, H.~Bergman, and P.~Brown}, {\em Pathological synchronization
  in {P}arkinson's disease: networks, models and treatments}, Trends Neurosci.,
  30 (2007), pp.~357--364.

\bibitem{HEB49}
{\sc D.~Hebb}, {\em {The Organization of Behavior: A Neuropsychological
  Theory}}, Wiley, New York, new edition~ed., 1949.

\bibitem{HOP96}
{\sc F.~C. Hoppensteadt and E.~M. Izhikevich}, {\em Synaptic organizations and
  dynamical properties of weakly connected neural oscillators {II}. {L}earning
  phase information}, Biol. Cybern., 75 (1996), pp.~129 --135,
  \url{https://doi.org/10.1007/s004220050280}.

\bibitem{JAI01}
{\sc S.~Jain and S.~Krishna}, {\em A model for the emergence of cooperation,
  interdependence, and structure in evolving networks}, Proc. Natl. Acad. Sci.,
  98 (2001), pp.~543--547, \url{https://doi.org/10.1073/pnas.98.2.543}.

\bibitem{KAS16a}
{\sc D.~V. Kasatkin and V.~I. Nekorkin}, {\em Dynamics of the phase oscillators
  with plastic couplings}, Radiophysics and Quantum Electronics, 58 (2016),
  pp.~877--891, \url{https://doi.org/10.1007/s11141-016-9662-1}.

\bibitem{KAS17}
{\sc D.~V. Kasatkin, S.~Yanchuk, E.~Sch{\"o}ll, and V.~I. Nekorkin}, {\em
  {S}elf-organized emergence of multi-layer structure and chimera states in
  dynamical networks with adaptive couplings}, Phys. Rev. E, 96 (2017), 062211
  (5~pages), \url{https://doi.org/10.1103/physreve.96.062211}.

\bibitem{KLO11}
{\sc P.~E. Kloeden and M.~Rasmussen}, {\em Nonautonomous {D}ynamical
  {S}ystems}, American Mathemaical Society, Providence, Rhode Island, 2011.

\bibitem{KUR84}
{\sc Y.~Kuramoto}, {\em Chemical Oscillations, Waves and Turbulence},
  Springer-Verlag, Berlin, 1984.

\bibitem{KUR02a}
{\sc Y.~Kuramoto and D.~Battogtokh}, {\em {Coexistence of Coherence and
  Incoherence in Nonlocally Coupled Phase Oscillators.}}, Nonlin. Phen. in
  Complex Sys., 5 (2002), pp.~380--385.

\bibitem{LEH14}
{\sc J.~Lehnert, P.~H{\"o}vel, A.~A. Selivanov, A.~L. Fradkov, and
  E.~Sch{\"o}ll}, {\em Controlling cluster synchronization by adapting the
  topology}, Phys. Rev. E, 90 (2014), 042914 (8~pages),
  \url{https://doi.org/10.1103/physreve.90.042914}.

\bibitem{LUE16}
{\sc L.~L{\"u}cken, O.~Popovych, P.~Tass, and S.~Yanchuk}, {\em
  {N}oise-enhanced coupling between two oscillators with long-term plasticity},
  Phys. Rev. E, 93 (2016), 032210 (15~pages).

\bibitem{LUE12a}
{\sc L.~L{\"u}cken and S.~Yanchuk}, {\em Two-cluster bifurcations in systems of
  globally pulse-coupled oscillators}, Physica D, 241 (2012), pp.~350--359,
  \url{https://doi.org/10.1016/j.physd.2011.10.017}.

\bibitem{MAI07}
{\sc Y.~Maistrenko, B.~Lysyansky, C.~Hauptmann, O.~Burylko, and P.~A. Tass},
  {\em Multistability in the kuramoto model with synaptic plasticity}, Phys.
  Rev. E, 75 (2007), 066207 (8~pages),
  \url{https://doi.org/10.1103/physreve.75.066207}.

\bibitem{MAR97a}
{\sc H.~Markram, J.~L\"ubke, and B.~Sakmann}, {\em Regulation of synaptic
  efficacy by coincidence of postsynaptic aps and epsps.}, Science, 275 (1997),
  pp.~213--215, \url{https://doi.org/10.1126/science.275.5297.213}.

\bibitem{NEK16}
{\sc V.~I. Nekorkin and D.~V. Kasatkin}, {\em Dynamics of a network of phase
  oscillators with plastic couplings}, AIP Conference Proceedings, 1738 (2016),
  210010 (4~pages), \url{https://doi.org/10.1063/1.4951993}.

\bibitem{OME11}
{\sc I.~Omelchenko, Y.~Maistrenko, P.~H{\"o}vel, and E.~Sch{\"o}ll}, {\em Loss
  of coherence in dynamical networks: spatial chaos and chimera states}, Phys.
  Rev. Lett., 106 (2011), 234102 (4~pages),
  \url{https://doi.org/10.1103/physrevlett.106.234102}.

\bibitem{OME12b}
{\sc O.~E. Omel'chenko and M.~Wolfrum}, {\em Nonuniversal transitions to
  synchrony in the sakaguchi-kuramoto model}, Phys. Rev. Lett., 109 (2012),
  164101 (4~pages), \url{https://doi.org/10.1103/physrevlett.109.164101}.

\bibitem{PER10c}
{\sc P.~Perlikowski, S.~Yanchuk, O.~Popovych, and P.~Tass}, {\em Periodic
  patterns in a ring of delay-coupled oscillators}, Phys. Rev. E, 82 (2010),
  036208 (12~pages), \url{https://doi.org/10.1103/physreve.82.036208}.

\bibitem{PIC11a}
{\sc C.~B. Picallo and H.~Riecke}, {\em Adaptive oscillator networks with
  conserved overall coupling: {S}equential firing and near-synchronized
  states}, Phys. Rev. E, 83 (2011), 036206 (12~pages),
  \url{https://doi.org/10.1103/physreve.83.036206}.

\bibitem{PIK08}
{\sc A.~Pikovsky and M.~G. Rosenblum}, {\em Partially integrable dynamics of
  hierarchical populations of coupled oscillators}, Phys. Rev. Lett., 101
  (2008), 264103 (4~pages),
  \url{https://doi.org/10.1103/physrevlett.101.264103}.

\bibitem{PIK01}
{\sc A.~Pikovsky, M.~G. Rosenblum, and J.~Kurths}, {\em Synchronization: a
  universal concept in nonlinear sciences}, Cambridge University Press,
  Cambridge, 2001.

\bibitem{POP13}
{\sc O.~Popovych, S.~Yanchuk, and P.~Tass}, {\em Self-organized noise
  resistance of oscillatory neural networks with spike timing-dependent
  plasticity}, Sci. Rep., 3 (2013), 2926 (6~pages).

\bibitem{POP15}
{\sc O.~V. Popovych, M.~N. Xenakis, and P.~A. Tass}, {\em The spacing principle
  for unlearning abnormal neuronal synchrony}, PLOS ONE, 10 (2015),
  \url{https://doi.org/10.1371/journal.pone.0117205}.

\bibitem{RAS06a}
{\sc C.~P\"{o}tzsche and M.~Rasmussen}, {\em Taylor approximation of integral
  manifolds}, Journal of Dynamics and Differential Equations, 18 (2006),
  pp.~427--460, \url{https://doi.org/10.1007/s10884-006-9011-8}.

\bibitem{REN07}
{\sc Q.~Ren and J.~Zhao}, {\em Adaptive coupling and enhanced synchronization
  in coupled phase oscillators}, Phys. Rev. E, 76 (2007), 016207 (6~pages),
  \url{https://doi.org/10.1103/physreve.76.016207}.

\bibitem{SAK86}
{\sc H.~Sakaguchi and Y.~Kuramoto}, {\em A soluble active rotater model showing
  phase transitions via mutual entertainment}, Prog. Theor. Phys, 76 (1986),
  pp.~576--581.

\bibitem{SEL02}
{\sc P.~Seliger, S.~C. Young, and L.~S. Tsimring}, {\em Plasticity and learning
  in a network of coupled phase oscillators}, Phys. Rev. E, 65 (2002), 041906
  (7~pages), \url{https://doi.org/10.1103/physreve.65.041906}.

\bibitem{SOR13}
{\sc M.~C. Soriano, J.~Garc{\'i}a-Ojalvo, C.~R. Mirasso, and I.~Fischer}, {\em
  Complex photonics: Dynamics and applications of delay-coupled semiconductors
  lasers}, Rev. Mod. Phys., 85 (2013), pp.~421--470.

\bibitem{STR00}
{\sc S.~H. Strogatz}, {\em From {K}uramoto to {C}rawford: exploring the onset
  of synchronization in populations of coupled oscillators}, Physica D, 143
  (2000), pp.~1--20.

\bibitem{TIM14}
{\sc L.~Timms and L.~Q. English}, {\em {S}ynchronization in phase-coupled
  {K}uramoto oscillator networks with axonal delay and synaptic plasticity},
  Phys. Rev. E, 89 (2014), 032906 (9~pages),
  \url{https://doi.org/10.1103/physreve.89.032906}.

\bibitem{TIN12}
{\sc M.~R. Tinsley, S.~Nkomo, and K.~Showalter}, {\em Chimera and phase cluster
  states in populations of coupled chemical oscillators}, Nat. Phys., 8 (2012),
  pp.~662--665, \url{https://doi.org/10.1038/nphys2371}.

\bibitem{VAN00}
{\sc V.~K. Vanag, L.~Yang, M.~Dolnik, A.~M. Zhabotinsky, and I.~R. Epstein},
  {\em Oscillatory cluster patterns in a homogeneous chemical system with
  global feedback}, Nature, 406 (2000), pp.~389--391,
  \url{https://doi.org/10.1038/35019038 10.1038/35019038}.

\bibitem{VER06}
{\sc F.~Verhulst}, {\em {Methods and applications of singular perturbations:
  boundary layers and multiple timescale dynamics}}, Springer, 2006.

\bibitem{WIL06}
{\sc D.~A. Wiley, S.~H. Strogatz, and M.~Girvan}, {\em The size of the sync
  basin}, Chaos, 16 (2006), 015103 (8~pages),
  \url{https://doi.org/10.1063/1.2165594}.

\end{thebibliography}
\end{document}